\documentclass[final,twocolumn]{IEEEtran}

\usepackage{amsmath,epsfig,amssymb,algorithm,amsthm,cite,url}

\usepackage{algpseudocode}
\usepackage{hyperref}

\newtheorem{theorem}{Theorem}
\newtheorem{lemma}[theorem]{Lemma}
\newtheorem{assumption}{Assumption~A-\kern-0pt}
\newtheorem{proposition}[theorem]{Proposition}
\newtheorem{corollary}[theorem]{Corollary}

\newtheorem*{remark*}{Remark}
\DeclareMathOperator{\tr}{tr}

\usepackage{pgfplots}

\DeclareMathOperator*{\argmax}{\arg\!\max}

\newcommand{\asto}{\overset{\rm a.s.}{\longrightarrow}}
\title{Optimal Design of the  Adaptive Normalized Matched Filter Detector}
\author{Abla Kammoun, Romain Couillet, Fr\'ed\'eric Pascal, Mohamed-Slim Alouini}
\begin{document}
\maketitle
\begin{abstract}
This article addresses improvements on the design of the adaptive normalized matched filter (ANMF)  for radar detection. It is well-acknowledged that the estimation of the noise-clutter covariance matrix  is a fundamental step in adaptive radar detection. In this paper, we consider regularized estimation methods which force by construction the eigenvalues of the  scatter estimates to be greater than a positive regularization parameter $\rho$. This makes them more suitable for high dimensional problems with a limited number of secondary data samples than traditional sample covariance estimates.
While an increase of $\rho$ seems to improve the conditioning of the estimate, it might however cause it to significantly deviate  from the true covariance matrix. The setting of the optimal regularization parameter is a difficult question for which no convincing answers have thus far been provided. This constitutes the major motivation behind our work. More specifically, we consider the design of the ANMF detector for two kinds of regularized estimators, namely  the regularized sample covariance matrix (RSCM), appropriate when the clutter follows a Gaussian distribution and the  regularized Tyler estimator (RTE) for non-Gaussian spherically invariant distributed clutters. The rationale behind this choice is that the RTE is efficient in mitigating the degradation caused by the presence of  impulsive noises while inducing little loss when the noise is Gaussian. 

Based on recent random matrix theory results studying the asymptotic fluctuations of the statistics of the ANMF detector when the number of samples and their dimension grow together to infinity, we propose a design for the regularization parameter that maximizes the detection probability under constant  false alarm rates. Simulation results which support the efficiency of the proposed method are  provided in order to illustrate the gain of the proposed optimal design over conventional settings of the regularization parameter.  



\end{abstract}
\section{Introduction}
The estimation of scatter matrices is of fundamental importance for space-time adaptive processing (STAP) which underlies the design  of radar systems \cite{ward-94,klemm-02}. In radar detection for instance, it is well-acknowledged  that a sufficiently accurate scatter matrix is key to enhancing the detection performance (see \cite{esa-12,mahot-13} and references therein). In order to support a possible deficiency in samples (number of samples less than their dimensions), regularized covariance matrix estimation methods have  been proposed \cite{abramovich-81}. One regularized estimation method is given by the use of the regularized sample covariance matrix (RSCM). The RSCM fundamentally originates from the diagonal loading approach which can be traced back to the works of Abramovich and Carlson \cite{abramovich-81,carlson-88}.
As a derivative of the sample covariance matrix (SCM), the RSCM inherits its main major limitation of exhibiting poor performances when observations contain outliers. This latter scenario is often modeled by assuming that observations are drawn from complex elliptical symmetric distributions (CES), originally introduced by Kelker \cite{kelker}. The inherent nature of these distributions to produce atypical observations makes the task of estimating the covariance matrix much more challenging. To tackle this  issue, a class of covariance  estimators coined robust estimators of scatter matrices have been proposed by Huber, Hampel, and Maronna \cite{huber1964robust,Huber72,Maronna76}, and extended more recently by Ollila to the complex case \cite{esa-12,mahot-13,pascal2008covariance}.    Similar to the Gaussian case, the regularization technique has been applied to the robust Tyler estimator \cite{tyler},  yielding the so-called regularized Tyler estimator (RTE). While conventional robust covariance methods  are undefined for too few samples (number of samples less than their dimensions), the existence of the RTE as well as the convergence of the associated recursive algorithm are two major findings which have recently been established in several works \cite{chitour-pascal-08,Pascal-2013,ollila-tyler,Palomar-14}. Unlike the RSCM, the RTE, as a derivative of the robust Tyler's estimator, is resilient to the presence of outliers, thereby making it more suitable to radar applications, for which experimental evidence rules out Gaussian models for the clutter \cite{Ward81,Watts85,Nohara91,Billingsley93}. 

As far as regularized estimation methods are concerned, it is essential to determine a clever way of setting the regularization parameter. This question has essentially been investigated in  \cite{wolf,chen-10} for the RSCM and in \cite{ollila-tyler,chen-11} for the RTE. Although yielding different expressions, these works have the common denominator of being merely based on a distance minimization between the RTE or the SCM and the true covariance matrix. It is thus not clear whether these choices will allow for good performances when applied to detection problems.  We consider in this work the design of the adaptive normalized matched filter (ANMF) for radar detection.
Firstly introduced by \cite{conte-95} and analyzed in \cite{liu2011acfar,kraut1999cfar,Pascal04-1}, this scheme was shown to enjoy the interesting features of constant false alarm property with respect to the clutter power and covariance matrix.
This detector is obtained by replacing in the statistic of the normalized matched filter (NMF) the covariance matrix by a given estimate \cite{kraut1999cfar}, which is computed based on secondary data observations, i.e., $n$ signal free independent and identically distributed (i.i.d.) observations. Of interest in this work are the cases where the RSCM or the RTE are used in place of the unknown covariance matrix. We will consider first the scenario where the detector operates over Gaussian correlated clutters and thus uses the RSCM as a replacement for the unknown covariance matrix, a scheme which  will be referred to as ANMF-RSCM. 
In order to come up with an appropriate design for the ANMF detector, it is essential to characterize the behaviour of its corresponding  false alarm and detection probabilities.  

Under the assumption of fixed dimensions, such a task seems to be out of reach.
This has led us to consider in \cite{couillet-kammoun-14} the regime wherein the
  number of secondary data samples and their dimensions grow simultaneously to infinity, thereby allowing for the use of advanced tools from random matrix theory. The asymptotic false alarm probability for the ANMF-RTE was in particular derived in \cite{couillet-kammoun-14} (yet only as a mere application example of the main mathematical result.). 
In order to allow for an optimal choice of the regularization parameter, these results need to be augmented with an asymptotic analysis of the detection probability. This constitutes the main contribution of our work.
In particular, we extend the results of \cite{couillet-kammoun-14} to the ANMF-RSCM detector by showing that its corresponding statistic flucutates as a Rayleigh distribution when no target is present and, additionally, establish that it behaves like a Rice distributed random variable otherwise.  Based on the asymptotic characterization of these distributions, we  propose an optimal setting of the regularization parameter that maximizes the asymptotic probability of detection for any given false alarm probability.

In a second part, we  consider the case where the clutter is drawn from heavy tailed distributions. It is thus natural to assume that the detector uses the RTE, since the RSCM is vulnerable to the presence of outliers and may provide poor performances. This scheme will be coined  ANMF-RTE.
By exploiting recent results on the asymptotic behavior of the RTE estimator \cite{couillet-kammoun-14,couillet-13}, we prove that, up to a certain change of variable, the  ANMF-RTE is asymptotically  equivalent to the ANMF-RSCM when operating over Gaussian clutters. 
This argues in favor of the role of the RTE to retrieve the Gaussian  performances while operating over heavy distributed clutters. 
Finally, we prove through simulations the  
superiority of our design  to some of the adhoc recent settings of the regularization parameter that have recently been proposed.  The  gain of our design likely  owes to the high accuracy of the derived asymptotic  results in predicting the detection probability of the ANMF schemes.


The remainder of the paper is organized as follows. In the first section, we introduce the considered problem. Then, we propose an optimal design approach for  the ANMF-RTE and the ANMF-RSCM.  Finally, we illustrate using simulations the gain of the proposed design method over conventional settings of the regularization parameter.

{\it Notations:} Throughout this paper, we depict vectors in lowercase boldface letters and matrices in uppercase boldface letters. The notation $(.)^*$ stands for the transpose conjugate  while $\tr(.)$ and $(.)^{-1}$ are the trace and inverse operators.  The notation $\|.\|$ stands for the Euclidean norm for vectors and for spectral norm for matrices. The arrow $\asto$ designates almost sure convergence. The statement $X\triangleq Y$ defines the new notation $X$ as being equal to $Y$.  
\section{Problem Statement}
We consider the problem of detecting a complex signal vector ${\bf p}$ corrupted by an additive noise as:
$$
{\bf y}=\alpha {\bf p}+{\bf x}
$$
where ${\bf y}\in\mathbb{C}^{N}$ represents the vector received by an  $N$-dimensional array of sensors, ${\bf x}$ stands for the noise clutter and $\alpha$ is a complex scalar modeling the unknown target amplitude. The signal detection problem is phrased as the following binary hypothesis test:
\begin{equation}
\left\{
\begin{array}{ll}
H_1:&{\bf y}=\alpha {\bf p} +{\bf x} \\
H_0:& {\bf y}={\bf x}.
\end{array}
\right.
\label{eq:hypothesis}
\end{equation}
Several models for the clutter ${\bf x}$ have been proposed.  Among them, we distinguish the class of  CES random variates which encompass most of the commonly encountered random models, including the standard Gaussian distribution, the K-distribution, the Weibull distribution and many others \cite{esa-12}. 
A CES distributed random variable is given by:
$$
{\bf x}=\sqrt{\tau}{\bf C}_N^{\frac{1}{2}}\tilde{\bf w}
$$
where ${\tau}$ is a positive scalar random variable called the {\it texture}, ${\bf C}_N$ is the covariance matrix\footnote{Note that when the second order statistics exist, the scatter matrix is equal to the covariance matrix (up to a constant).} and $\tilde{\bf w}$ is an $N$-dimensional vector independent of $\tau$, zero-mean unitarily invariant  with norm $\|\tilde{\bf w}\|=\sqrt{N}$. The quantity ${\bf C}_N^{\frac{1}{2}}\tilde{\bf w}$  is referred to as {\it speckle}.  The design of an appropriate statistic to the above hypothesis test depends on the amount of knowledge that is available to the detector. 
If the clutter is Gaussian with known covariance matrix ${\bf C}_N$   while $\alpha$ is unknown, the Generalized Likelihood Ratio (GLRT) for the detection problem in \eqref{eq:hypothesis} results in the following test statistic:
$$
T_N=\frac{\left|{\bf y}^{*}{\bf C}_N^{-1}{\bf p}\right|}{\sqrt{{\bf y}^{*}{\bf C}_N^{-1}{\bf y}{\bf p}^{*}{\bf C}_N^{-1}{\bf p}}} 
$$
which corresponds to the square-root statistic of the ANMF detector. 
The statistic $T_N$ has been derived independently by several works, thereby leading the corresponding detector to have  many alternative names: the constant false alarm (CFAR) matched subspace detector (MSD) \cite{scharf-94}, the normalized matched filter (NMF) \cite{conte-02}, or the Linear Quadratic GLRT (LQ-GLRT) \cite{Gini-02}. If the clutter is elliptically distributed, optimal detection procedures based on the GLRT principle lead to statistics that depend on the distribution of the texture $\tau$.
Nevertheless, a complete knowledge of the statistics of the signal and noise cannot be acquired in practice. A reasonable hypothesis, largely used  in radar detection, is to assume that only ${\bf p}$ is known while $\alpha$ and the statistics of the noise are
ignored. To handle this case, the use of $T_N$ whose optimality has only been shown in the Gaussian setting has been advocated as a good detection technique. Such a choice has particularly been motivated by the result of \cite{conte-95} showing the asymptotic optimality of $T_N$, when $N$ becomes increasingly large, under the setting of compound-Gaussian distributed clutters. From the expression of $T_N$, it can be seen that the detector is only required to know ${\bf C}_N$ up to a scale factor, which is much less restrictive than the requirement of optimal detection strategies.

Since the covariance matrix ${\bf C}_N$ is unknown in practice, a popular approach consists in replacing in ${ T}_N$ the unknown covariance matrix ${\bf C}_N$ by an estimate built on signal free i.i.d. observations ${\bf x}_1,\cdots,{\bf x}_n$, termed  secondary data. The resulting detector is called  the adaptive normalized matched filter (ANMF). Several concurrent estimators of ${\bf C}_N$ can be used. The most popular one is the traditional sample covariance matrix (SCM) given by:
$$
\widehat{\bf R}_N=\frac{1}{n}\sum_{i=1}^n {\bf x}_i{\bf x}_i^*
$$
which corresponds to the Maximum-Likelihood estimator (MLE) if the clutter is Gaussian distributed. However, in some scenarios where the available number of observations $n$ is of the same order or smaller than $N$,  the SCM, being ill-conditioned, will not lead to accurate detection results\footnote{Traditionally, it is assumed that $2N$ observations are required to ensure good performances of the sub-optimal filtering, i.e., a $3$ dB loss of the output SNR compared to optimal filtering \cite{Reed-74}.}. A practical  approach that has received considerable attention is to regularize the SCM, thereby yielding the regularized SCM (RSCM) given by:
 \begin{equation}
\widehat{\bf R}_N(\rho)=(1-\rho)\widehat{\bf R}_N+\rho{\bf I}_N,
\label{eq:R_scm}
\end{equation} where the parameter $\rho\in\left[0,1\right]$ serves to give more or less importance to the sample covariance matrix $\widehat{\bf R}_N$ depending on the available number of samples. The ANMF that uses the RSCM as a plug-in estimator of ${\bf C}_N$ will be referred to as ANMF-RSCM.

While this regularization artifice has revealed  efficient in handling the scarcity of the available samples, it has the serious drawback of fundamentally relying on the SCM. In effect, the SCM, even though suitable for Gaussian settings, is known to be vulnerable to outliers and thus leads to highly inefficient estimators when the samples are drawn from heavy tailed non-Gaussian distributions. 
 A standard alternative to conventional sample covariance estimates is constituted by the class of robust-scatter estimators, known for their resilience to atypical observations. The robust estimator that will be considered in this work was defined in \cite{Pascal-2013} as the unique solution $\hat{\bf C}_N(\rho)$ to:
 \begin{equation}
\hat{\bf C}_N(\rho)=(1-\rho)\frac{1}{n}\sum_{i=1}^n \frac{{\bf x}_i{\bf x}_i^*}{\frac{1}{N}{\bf x}_i^*\hat{\bf C}_N^{-1}(\rho){\bf x}_i} + \rho {\bf I}_N.
\label{eq:robust}
\end{equation}
with $\rho\in\left(\max\left(0,1-\frac{n}{N}\right),1\right]$. 
This estimator corresponds to a hybrid robust-shrinkage estimator reminding Tyler's M-estimator of scale \cite{tyler} and Ledoit-Wolf's shrinkage estimator \cite{wolf}. We will thus refer to it as the Regularized-Tyler Estimator (RTE). Besides its robustness, the RTE has many interesting features. First,  it is well-suited to situations where $c_N\triangleq\frac{N}{n}$ is large while standard robust scatter estimates are  ill-conditioned or even undefined if $N> n$. By varying the regularization parameter $\rho$, one can move from the  unbiased Tyler-estimator \cite{Pascal-08} ($\rho=0$) to the identity matrix $(\rho=1)$ which represents a crude guess for the unknown covariance ${\bf C}_N$. Its relation to the Tyler's estimator has recently been reported in \cite{ollila-tyler} by viewing it as the solution of a penalized $M$-estimation cost function. We will denote by ANMF-RTE the ANMF detector that uses the RTE instead of the unknown covariance matrix. 

Upon replacing in  $T_N$ the unknown covariance matrix by a regularized estimate, be it the SCM or the RTE, the question of how should the regularization parameter $\rho$ be set naturally arises. Recent previous works dealing with this issue propose to set $\rho$ in such a way as to minimize a certain mean-squared-error  between $\hat{\bf C}_N$ and ${\bf C}_N$ \cite{ollila-tyler,abramovich-13}. While easy-to-compute estimates of these values of $\rho$ were provided, one of the major criticism to these choices is that they are performed regardless of the application under consideration. In particular, a more relevant choice to the application under study consists in selecting the values of $\rho$ that maximize the probability of detection while keeping fixed the false alarm probabilities. These values will be considered as optimal in regards of radar detection applications.

To this end, one needs to characterize the distribution of $\widehat{T}_N^{\rm RSCM}(\rho)$ and $\widehat{T}_N^{\rm RTE}(\rho)$ given by:
\begin{align}
\widehat{T}_N^{\rm RSCM}(\rho)&=\frac{\left|{\bf y}^{*}\widehat{\bf R}_N^{-1}(\rho){\bf p}\right|}{\sqrt{{\bf y}^{*}\widehat{\bf R}_N^{-1}(\rho){\bf y}}\sqrt{{\bf p}^{*}\widehat{\bf R}_N^{-1}(\rho){\bf p}}} \label{eq:scm}\\
\widehat{T}_N^{\rm RTE}(\rho)&=\frac{\left|{\bf y}^{*}\widehat{\bf C}_N^{-1}(\rho){\bf p}\right|}{\sqrt{{\bf y}^{*}\widehat{\bf C}_N^{-1}(\rho){\bf y}}\sqrt{{\bf p}^{*}\widehat{\bf C}_N^{-1}(\rho){\bf p}}} \label{eq:rte}
\end{align}
under hypotheses $H_0$ and $H_1$.
For fixed $N$ and $n$, this is not an easy task and in our opinion would not lead, if ever feasible, to easy-to-compute expressions for the optimal values of $\rho$. In this paper, we relax this restrictive assumption by considering the case where $N$ and $n$ go to infinity with $\frac{N}{n}\to c\in(0,\infty)$. This in particular enables  leveraging the recent results of \cite{couillet-kammoun-14} that will be reviewed in Section \ref{sec:background}.

\section{Optimal design of the ANMF-RSCM detector: Gaussian clutter case}
In this section, we consider the case of a Gaussian clutter. In other words, we assume that all the secondary data ${\bf x}_1,\cdots,{\bf x}_n$ are drawn from Gaussian distribution with zero-mean and covariance ${\bf C}_N$. For $\rho\in\left(0,1\right]$, we define the RSCM as in \eqref{eq:R_scm} and corresponding statistic $\widehat{T}_N^{RSCM}$. In order to pave the way towards an optimal setting of the regularization coefficient $\rho$, we need to characterize the asymptotic false alarm and detection probabilities under the assumptions that $c_N\triangleq \frac{N}{n}\to c$. That is, provided $H_0$ or $H_1$ is the actual scenario, $({\bf y}={\bf x}$ or ${\bf y}=\alpha{\bf p}+{\bf x}$), we shall evaluate the probabilities $\mathbb{P}\left[\widehat{T}_N^{\rm RSCM}> \Gamma| H_0\right]$ and  $\mathbb{P}\left[\widehat{T}_N^{\rm RSCM}> \Gamma| H_1\right]$ for $\Gamma >0$. Before going further, we need to stress that some extra assumptions on the order of magnitude of $\alpha$ and $\Gamma$ with respect to $N$ should be made to avoid getting trivial results. Indeed, it appears that under $H_0$, the random quantities $\frac{1}{\sqrt{N}}{\bf y}^*\widehat{\bf R}_N^{-1}(\rho)\frac{\bf p}{\|{\bf p}\|}$, $\frac{1}{N}{\bf y}^{*}\widehat{\bf R}_N^{-1}(\rho){\bf y}$, and ${\bf p}^*\widehat{\bf R}_N^{-1}(\rho)\frac{\bf p}{\|{\bf p}\|^2}$ are standard objects in random matrix theory, which converge almost surely to their  means when both $N$ and $n$ grow to infinity with the same pace\cite{WAG10}. As a result, since $\frac{1}{\sqrt{N}}{\bf y}^*\widehat{\bf R}_N^{-1}(\rho){\bf p}\asto 0$,   $\widehat{T}_N^{\rm RSCM}\asto 0$ for all $\Gamma >0$, which does not allow to infer much information about the false alarm probability. 
It turns out that the proper scaling of $\Gamma$ should be $\Gamma=N^{-\frac{1}{2}} r$ for some fixed $r>0$, an assumption already considered in \cite{couillet-kammoun-14}. Similarly, one can see that under $H_1$, the presence of a signal component in ${\bf y}$ causes  $\widehat{T}_N^{\rm RSCM}$ to converge almost surely to some positive constant if $\alpha$ does not vary with $N$. Therefore, for $\Gamma=N^{-\frac{1}{2}} r$, $\mathbb{P}\left[\widehat{T}_N^{\rm RSCM}> \Gamma| H_1\right]\to 1$. In order to avoid this trivial statement, we shall assume that $\alpha = N^{-\frac{1}{2}}a$ for some fixed $a>0$ with $\|{\bf p}\|=N$. In practice, this means that the dimension of the array is sufficiently large to enable working in low-SNR regimes.

Prior to introducing the results about the false alarm and detection probabilities, we shall introduce the following assumptions and notations:
\begin{assumption}
For $i\in\left\{1,\cdots,n\right\}$, ${\bf x}_i={\bf C}_N^{\frac{1}{2}}{\bf w}_i$ with:
\begin{itemize}
\item ${\bf w}_1,\cdots,{\bf w}_n$ are $N\times 1$ independent standard Gaussian random vectors with zero-mean and covariance ${\bf I}_N$,
\item ${\bf C}_N \in\mathbb{C}^{N\times N}$ is such that 
$\lim\sup \|{\bf C}_N\| < \infty$ and $\frac{1}{N}\tr{\bf C}_N=1$,
\item 
$
\lim\inf_N \frac{1}{N}{\bf p}^*{\bf C}_N{\bf p} >0.
$
\end{itemize}
\label{ass:gaussian}
\end{assumption}
Note that the normalization $\frac{1}{N}\tr{\bf C}_N=1$ is not a restricting constraint since the statistics under study are invariant to the scaling of ${\bf C}_N$. The last item in Assumption \ref{ass:gaussian} is required for technical purposes in order to ensure that the considered statistic exhibits fluctuations under $H_0$ and $H_1$. In practice, this assumption implies that the steering vector does not lie in the null space of the  covariance matrix ${\bf C}_N$.

Denote for $z\in\mathbb{C}\backslash\mathbb{R}_{+}$ by $m_N(z)$ the unique complex solution  to:
\begin{align*}
m_N(z)&=\left(-z+c_N(1-\rho)\right.\\
&\left.\times\frac{1}{N}\tr {\bf C}_N\left({\bf I}_N+(1-\rho)m_N(z){\bf C}_N\right)^{-1} \right)^{-1}
\end{align*}
that satisfies $\Im(z)\Im(m_N(z))\geq 0$ or unique positive if $z<0$.
The existence and uniqueness of $m_N(z)$ follows from standard results of random matrix theory \cite{SIL95}. 
It is a deterministic quantity, which can be computed easily for each $z$ using fixed-point iterations. In our case, it helps characterize the asymptotic behavior of the empirical spectral measure of the random matrix ${(1-\rho)}\frac{1}{n}\sum_{i=1}^n {\bf x}_i{\bf x}_i^*$\footnote{Let $\hat{\nu}_N=\frac{1}{N}\sum_{i=1}^N \delta_{\lambda_i}$ be the empirical spectral measure of the random matrix  $\frac{1-\rho}{n}\sum_{i=1}^n {\bf x}_i{\bf x}_i^*$ with $\lambda_1,\cdots,\lambda_N$  the eigenvalues of $\frac{1-\rho}{n}\sum_{i=1}^n {\bf x}_i{\bf x}_i^*$. 
Denote by $\hat{m}_N(z)$ its Stieltjes transform given by $\hat{m}_N(z)=\int (t-z)^{-1}\hat{\nu}_N(dt)=\frac{1}{N}\sum_{i=1}^N \frac{1}{\lambda_i-z}$.
Then, quantity $m_N(z)$ is the Stieljes transform of a certain deterministic measure $\mu_N$, (i.e, $m_N(z)=\int (t-z)^{-1}\mu_N(dt)$) which approximates in the almost sure sense $\hat{m}_N(z)$ (i.e., $\hat{m}_N(z)-m_N(z)\asto 0.$).}.

Define also for $\kappa>0$, $\mathcal{R}_{\kappa}^{\rm SCM}$ as:
$$
\mathcal{R}_{\kappa}^{\rm SCM}\triangleq \left[\kappa,1\right].
$$

With these notations at hand, we are now ready to analyze the asymptotic behaviour of the false alarm and detection probabilities. 
The proof for the following Theorem will not be provided since, as we shall see in Section \ref{sec:non_gaussian_clutter}, it  follows directly by applying the same approach used in \cite{couillet-kammoun-14}.

\begin{theorem}[False alarm probability]
As $N,n\to \infty$ with $c_N\to c\in(0,\infty)$,
\begin{align*}
\sup_{\rho\in\mathcal{R}_{\kappa}^{\rm SCM}}\left|\mathbb{P}\left[\widehat{T}_N^{\rm RSCM}(\rho)> \frac{r}{\sqrt{N}}|H_0\right]-e^{-\frac{r^2}{2\sigma_{N,{\rm SCM}}^2(\rho)}}\right| \asto 0
\end{align*}
where:
\begin{align*}
\sigma_{N,{\rm SCM}}^2(\rho)&\triangleq \frac{1}{2}\frac{{\bf p}^*{\bf C}_N{\bf Q}_N^2(\rho){\bf p}}{{\bf p}^*{\bf Q}_N(\rho){\bf p}\frac{1}{N}\tr {\bf C}_N{\bf Q}_N(\rho)} \\
&\times \frac{1}{1-c(1-\rho)^2m_N(-\rho)^2\frac{1}{N}\tr {\bf C}_N^2{\bf Q}_N^2(\rho)}
\end{align*}
\label{th:false_alarm}
and ${\bf Q}_N(\rho)\triangleq \left({\bf I}_N+(1-\rho)m_N(-\rho){\bf C}_N\right)^{-1}$.
\end{theorem}
The uniformity over $\rho$ of the convergence result in Theorem~\ref{th:false_alarm} is essential in the sequel.
It obviously implies the pointwise convergence for each $\rho >0$ but, more importantly, it will allow us to handle the convergence of the false alarm probability when random values of the regularization parameter are considered. This feature  becomes all the more interesting knowing that the detector is required to set the regularization parameter based on  random received secondary data. Note that, for technical issues, a set of the form $\left[0,\kappa\right)$, where $\kappa>0$ is as small as desired but fixed, has to be discarded from the uniform convergence region. 

The result of Theorem \ref{th:false_alarm} provides an analytical expression for the false alarm probability. Since this expression  depends on the unknown covariance matrix, it is of practical interest to provide a consistent estimate for it:
\begin{proposition}
For $\rho\in(0,1)$, define
$$
\hat{\sigma}_{N,{\rm SCM}}^2(\rho)=\frac{1}{2}\frac{1-\rho\frac{{\bf p}^*\widehat{\bf R}_{N}^{-2}(\rho){\bf p}}{{\bf p}^*\widehat{\bf R}_{N}^{-1}{\bf p}}}{\left(1-c_N+\frac{c_N\rho}{N}\tr \widehat{\bf R}_N^{-1}(\rho)\right)\left(1-\frac{\rho}{N}\tr \widehat{\bf R}_N^{-1}(\rho)\right)}
$$
and let $\hat{\sigma}_{N,SCM}^2(1)=\lim_{\rho\uparrow 1} \hat{\sigma}_{N,SCM}^2(\rho)=\frac{{\bf p}^*\widehat{\bf R}_N{\bf p}}{\tr \widehat{\bf R}_N}$. Then, we have, for any $\kappa>0$,
$$
\sup_{\rho\in\mathcal{R}_{\kappa}^{{\rm SCM}}}\left|\hat{\sigma}_{N,{\rm SCM}}^2(\rho)-\sigma_{N,{\rm SCM}}^2(\rho)\right|\asto 0.
$$
\label{prop:estimation_variance}
\end{proposition}
The proof of Proposition \ref{prop:estimation_variance} follows along the same lines as that of Proposition 1 in \cite{couillet-kammoun-14}  and is therefore omitted.

We will now derive the asymptotic equivalent for $\mathbb{P}\left[\widehat{T}_N^{\rm RSCM}(\rho)> \frac{r}{\sqrt{N}}|H_1\right]$, where under $H_1$ the received vector ${\bf y}$ is supposed to be given by:
$$
H_1: \hspace{0.2cm}{\bf y}=\frac{a}{\sqrt{N}}{\bf p}+{\bf x}
$$
with ${\bf x}$  distributed as the ${\bf x}_i$'s in Assumption \ref{ass:gaussian}.
The following results constitute the major contribution of the present work. They will lead in conjunction with those of Theorem \ref{th:false_alarm_1} and Proposition \ref{prop:estimation_variance} to the optimal design of the ANMF-RSCM. 

\begin{theorem}[Detection probability]
As $N,n\to\infty$ with $c_N\to c$, we have for any $\kappa>0$
\begin{align*}
&\sup_{\rho\in\mathcal{R}_{\kappa}^{\rm SCM}}\left|\mathbb{P}\left[\widehat{T}_N^{\rm RSCM}(\rho)> \frac{r}{\sqrt{N}}|H_1\right]\right.\\
&\left.-Q_1\left(g_{\rm SCM}({\bf p}),\frac{r}{\sigma_{N,{\rm SCM}}(\rho)}\right)\right| \asto 0.
\end{align*}
where $Q_1$ is the Marcum Q-function\footnote{ $Q_1(a,b)=\int_b^{+\infty}x \exp\left(-\frac{x^2+a^2}{2}\right)I_0(ax)dx$ where $I_0$ is the zero-th order modified Bessel function of the first kind.} while $\sigma_{N,{\rm SCM}}$ is given in Theorem \ref{th:false_alarm} and ${g}_{SCM}({\bf p})$ is given by:
\begin{align*}
{ g}_{SCM}({\bf p})&=\frac{\sqrt{1-c(1-\rho)^2m(-\rho)^2\frac{1}{N}\tr{\bf C}_N^2{\bf Q}_N^2(\rho)}}{\sqrt{{\bf p}^*{\bf C}_N{\bf Q}_N^2(\rho){\bf p}}} \\
&\times\sqrt{\frac{2}{N}}a{\left|{\bf p}^*{\bf Q}_N(\rho){\bf p}\right|}.
\end{align*}
\label{th:detection}
\end{theorem}
\begin{proof}
See Appendix \ref{app:detection}.
\end{proof}
According to Theorem \ref{th:false_alarm} and Theorem \ref{th:detection}, $\widehat{T}_N^{\rm RSCM}(\rho)$ behaves differently depending on whether a  signal is present or not. In particular, under $H_0$,  $\sqrt{N}\widehat{T}_N^{\rm RSCM}(\rho)$ behaves like a Rayleigh distributed random variate with parameter $\sigma_{N,{\rm SCM}}(\rho)$ while it becomes well-approximated  under $H_1$ by a Rice distributed random variable with parameters $g_{\rm SCM}({\bf p})$ and $\sigma_{N,{\rm SCM}}(\rho)$. 
It is worth noticing that in the theory of radar detection, getting a false alarm and a detection probability  distributed as Rayleigh and Rice random variables is among the  simplest cases  that one can ever encounter, holding only, to the best of the authors' knoweldege,  if white Gaussian noises are considered \cite[p.188]{Kay93}.
 We believe that the striking simplicity of the obtained results inheres in the double averaging effect that is a consequence of the considered asymptotic regime. This is to be compared to the quite intricate expressions for the false alarm probability obtained under the classical regime of $n$ tending to infinity while $N$ is fixed \cite{pascal-icassp15}.  

We will now discuss the choice of the regularization parameter $\rho$ and the threshold $r$. In accordance with the theory of radar detection, we aim at setting $\rho$ and $r$ in such a way to keep the asymptotic false alarm probability equal to a fixed value $\eta$ while maximizing the asymptotic probability of detection.
From Theorem \ref{th:false_alarm}, one can easily see that the values of $r$ and $\rho$ that provide an asymptotic false alarm probability equal to $\eta$ should satisfy:
$$
\frac{r}{\sigma_{N,{\rm SCM}}(\rho)}=\sqrt{-2\log\eta}.
$$
From these choices, we have to take those values that maximize the asymptotic detection which is  given, according to Theorem \ref{th:detection}, by:
$$
Q_1\left({g_{\rm SCM}({\bf p})}, \frac{r}{\sigma_{N,SCM}(\rho)}\right).
$$
 The second argument of $Q_1$ should be kept fixed in order to ensure the required asymptotic false alarm probability. As the Marcum-Q function increases with respect to the first argument, the optimization of the detection probability boils down to considering the following values of $\rho$:
$$
\rho\in\argmax \left\{f_{\rm SCM}(\rho)\right\}
$$
where:
$$
f_{\rm SCM}(\rho)\triangleq \frac{1}{2a^2}g_{\rm SCM}^2\left({\bf p}\right)
$$
However, the optimization of  $f_{\rm SCM}(\rho)$ is not possible in practice, since the expression of $f_{SCM}(\rho)$ features the covariance matrix ${\bf C}_N$ which is unknown to the detector. Acquiring a consistent estimate of $f_{\rm SCM}(\rho)$ based on the available $\widehat{\bf R}_N$ is thus mandatory. This is the goal of the following Proposition.
\begin{proposition}
For $\rho\in\left(0,1\right)$, define $\hat{f}_{\rm SCM}(\rho)$ as:
$$
\hat{f}_{\rm SCM}(\rho)=\frac{\left({\bf p}^*\widehat{\bf R}_N^{-1}(\rho){\bf p}\right)^2(1-\rho)\left(1-c+\frac{c}{N}\rho\tr\widehat{\bf R}_N^{-1}(\rho)\right)^2}{{\bf p}^*\widehat{\bf R}_N^{-1}(\rho){\bf p}-\rho{\bf p}^*\widehat{\bf R}_N^{-2}(\rho){\bf p}}
$$
and let $\hat{f}_{\rm SCM}(1)\triangleq \lim_{\rho\uparrow 1}\hat{f}_{\rm SCM}(\rho)=\frac{N}{{\bf p}^*\widehat{\bf R}_N{\bf p}}$. Then, we have:
$$
\sup_{\rho\in\mathcal{R}_{\kappa}^{\rm SCM}} \left|\hat{f}_{\rm SCM}(\rho)-f_{\rm SCM}(\rho)\right|\asto 0,
$$
where we recall that $\mathcal{R}_{\kappa}^{\rm SCM}=\left[\kappa,1\right]$.
\label{prop:frho}
\end{proposition}
\begin{proof}
See Appendix \ref{app:frho}.
\end{proof}
Since the results in Proposition \ref{prop:frho} and Theorem \ref{th:detection} are uniform in $\rho$, we have the following corollary:
\begin{corollary}
 Let $\hat{f}_{SCM}(\rho)$ be defined as in Proposition \ref{prop:frho}. Define $\hat{\rho}_N^*$ as any value satisfying:
$$
\hat{\rho}_N^*\in\argmax_{\rho\in\mathcal{R}_{\kappa}^{\rm SCM}}\left\{\hat{f}_{\rm SCM}(\rho)\right\}.
$$
Then, for every $r>0$,
\begin{align*}
&\mathbb{P}\left(\sqrt{N}T_N(\hat{\rho}_N^*) >r |H_1\right)\\
&-\max_{\rho\in\mathcal{R}_{\kappa}^{\rm SCM}}\left\{\mathbb{P}\left(\sqrt{N}T_N({\rho})>r|H_1\right)\right\}\asto0.
\end{align*}
\label{cor:optimal}
\end{corollary}
\begin{proof}
The proof is similar to that of Corollary 1 of \cite{couillet-kammoun-14} and is thus omitted.
\end{proof}
From Corollary \ref{cor:optimal}, the following design procedure leads to optimal performance detection results:
\begin{itemize}
\item First, setting the regularization parameter to one of the values maximizing $\hat{f}_{\rm SCM}(\rho)$:
\begin{equation}
\hat{\rho}_N^* \in\argmax_{\rho\in\mathcal{R}_{\kappa}^{\rm SCM}}\left\{\hat{f}_{\rm SCM}(\rho)\right\}
\label{eq:rho_N}
\end{equation}
\item Second, selecting the threshold $\hat{r}$ as:
\begin{equation}
\hat{r}=\hat{\sigma}_{N,{\rm SCM}}(\hat{\rho}_N^*) \sqrt{-2\log\eta}
\label{eq:threshold}
\end{equation}
\end{itemize}
\section{Optimal design of the ANMF-RTE: Non-Gaussian clutter}
\label{sec:non_gaussian_clutter}
This section discusses the design of the ANMF-RTE detector in the case where the clutter is non-Gaussian. In particular, we assume that the secondary observations satisfy the following assumptions:
\begin{assumption}
For $i\in\left\{1,\cdots,n\right\}$, ${\bf x}_i=\sqrt{\tau_i}{\bf C}_N^{\frac{1}{2}} {\bf w}_i=\sqrt{\tau_i}{\bf z}_i$ where
\begin{itemize}
	\item ${\bf w}_1,\cdots,{\bf w}_n$ are $N\times 1$  independent unitarly invariant complex   zero-mean random vectors with $\|{\bf w}_i\|^2=N$,
\item ${\bf C}_N \in \mathbb{C}^{N\times N}$ is such that $\lim\sup \|C_N\| < \infty$ and $\frac{1}{N}\tr{\bf C}_N=1$.  
\item $\tau_i>0$ are independent of ${\bf w}_i$.
\item $\lim\inf \frac{1}{N}{\bf p}^*{\bf C}_N{\bf p} >0$.
\end{itemize}
\label{ass:model_elliptical}
\end{assumption}
The random model described in Assumption \ref{ass:model_elliptical} is that of CES distributions which encompass a wide range of observation distributions obtained for different settings of the statistics of $\tau_i$. 
Prior to stating our main findings, we shall first review   some recent results concerning the asymptotic behaviour of the RTE in the asymptotic regime.
\subsection{Background}
This section reviews the recent results in \cite{couillet-kammoun-14} about the asymptotic behaviour of the RTE estimator.

Recall that the RTE is defined,  for $\rho\in\left(\max\left\{0,1-\frac{n}{N}\right\},1\right]$, as the unique solution to the following equation:
$$
\hat{\bf C}_N(\rho)=(1-\rho)\frac{1}{n}\sum_{i=1}^n \frac{{\bf x}_i{\bf x}_i^*}{\frac{1}{N}{\bf x}_i^*\hat{\bf C}_N^{-1}(\rho){\bf x}_i} +\rho {\bf I}_N.
$$
The study of the asymptotic behaviour of robust-scatter estimators is much more challenging than that of the traditional sample covariance matrices. The main reasons are that, first, robust estimators of scatter do not have closed-form expressions and, second, the dependence between the outer-products involved in their expressions is non-linear, which does not allow for the use of standard random matrix analysis. In order to study this class of estimators, new technical tools based on different rewriting of the robust-scatter estimators have been developed by Couillet et al. \cite{couillet-pascal-2013,couillet-13,couillet-13a}. The important advantage of these techniques is that they suggest to replace robust estimators by asymptotically equivalent random matrices for which many results from random matrix theory are applicable. In particular, the RTE estimator defined above has been studied in \cite{couillet-kammoun-14} and has been shown to behave in the regime where $N,n\to\infty$ in such a way that $c_N\to c\in(0,\infty)$ similar to $\hat{\bf S}_N(\rho)$ given by:
\begin{equation}
\hat{\bf S}_N(\rho)=\frac{1}{\gamma_N(\rho)}\frac{1-\rho}{1-(1-\rho)c_N}\frac{1}{n}\sum_{i=1}^n {\bf z}_i{\bf z}_i^{*}+\rho {\bf I}_N,
\label{eq:hatSN}
\end{equation}
where $\gamma_N(\rho)$ is the unique solution to:
$$
1=\int \frac{t}{\gamma_N(\rho)\rho+(1-\rho)t} \nu_N(dt).
$$
More specifically, the following theorem applies:
\begin{theorem}[\cite{couillet-13}]
For any $\kappa>0$ small, define $\mathcal{R}_\kappa^{\rm RTE}\triangleq \left[\kappa+\max(0,1-c^{-1}),1\right]$. Then,  as $N,n\to\infty$ with $c_N\to c\in(0,\infty)$, we have:
$$
\sup_{\rho\in\mathcal{R}_\kappa^{\rm RTE}}\left\|\hat{\bf C}_N(\rho)-\hat{\bf S}_N(\rho)\right\|\asto 0.
$$
\label{th:first_order}
\end{theorem}
Theorem \ref{th:first_order} establishes a convergence in the operator norm of the difference $\hat{\bf C}_N(\rho)-\hat{\bf S}_N(\rho)$.  This  result allows one to transfer the asymptotic first order analysis of many functionals of $\hat{\bf C}_N(\rho)$ to  $\hat{\bf S}_N(\rho)$. However, when it comes to the study of  fluctuations, this result is of little help. Indeed, although Theorem \ref{th:first_order} can be easily refined as
$$
\sup_{\rho\in\mathcal{R}_\kappa^{RTE}}N^{\frac{1}{2}-\epsilon}\left\|\hat{\bf C}_N(\rho)-\hat{\bf S}_N(\rho)\right\|\asto 0.
$$
for each $\epsilon >0$, the above convergence does not suffice to obtain  the convergence of most of the commonly used functionals which involve fluctuations of order $N^{-\frac{1}{2}}$ or $N^{-1}$ (e.g. quadratic forms of $\hat{\bf C}_N(\rho)$ or linear statistics of the eigenvalues of $\hat{\bf C}_N(\rho)$). While a further refinement of the above convergence seems to be out of reach, it has recently been established  in \cite{couillet-kammoun-14} that the fluctuations of special functionals can be proved to be much faster, mainly by virtue of an averaging effect which cancels out terms fluctuating at lower speed. In particular, bilinear forms of the type ${\bf a}^{*}\hat{\bf C}_N^{k}(\rho){\bf b}$ were studied in \cite{couillet-kammoun-14}, where the following proposition was proved:
\begin{proposition}
	Let ${\bf a},{\bf b}\in \mathbb{C}^{N}$ with $\|{\bf a}\|=\|{\bf b}\|=1$ deterministic or random independent of ${\bf x}_1,\cdots,{\bf x}_n$. Then, as $N,n\to\infty$, with $c_N\to c \in(0,\infty)$, for any $\epsilon >0$ and every $k\in \mathbb{Z}$, 
$$
\sup_{\rho\in\mathcal{R}_\kappa^{RTE}}N^{1-\epsilon}\left|{\bf a}^*\hat{\bf C}_N^k(\rho){\bf b}-{\bf a}^*\hat{\bf S}_N^k(\rho){\bf b}\right|\asto 0.
$$
where $\mathcal{R}_\kappa^{RTE}$ is defined as in Theorem \ref{th:first_order}, where $k\in\mathbb{Z}$ in any power of the matrices $\hat{\bf C}_N$ and $\hat{\bf S}_N$. 
\label{th:second_order}
\end{proposition}
Some important consequences of Proposition \ref{th:second_order} need to be stated. First, we shall recall that, while the crude study of the random variates ${\bf a}^{*}\hat{\bf C}_N^k(\rho){\bf b}$ seems to be intractable, quadratic forms of the type ${\bf a}^{*}\hat{\bf S}_N^k(\rho){\bf b}$ are well-understood objects whose behavior can be studied using standard tools from random matrix theory \cite{Kammoun09}. It is thus interesting to transfer the study of the fluctuations of  ${\bf a}^{*}\hat{\bf C}_N^k(\rho) {\bf b}$ to  ${\bf a}^{*}\hat{\bf S}_N^k(\rho) {\bf b}$. Proposition \ref{th:second_order} achieves this goal by taking $\epsilon <\frac{1}{2}$. Not only does it entail that ${\bf a}^{*}\hat{\bf C}_N^{k}(\rho){\bf b}$ fluctuates at the order of $N^{-\frac{1}{2}}$ (since so does ${\bf a}^*\hat{\bf S}_N^{k}(\rho){\bf b}$) but also it allows one to prove that ${\bf a}^{*}\hat{\bf C}_N^{k}(\rho){\bf b}$ and ${\bf a}^{*}\hat{\bf S}_N^k(\rho) {\bf b}$ exhibit asymptotically the same fluctuations. 
Similar to \cite{couillet-kammoun-14}, our concern will be rather focused on the case  $k=-1$.  In the next section, we will show how this result can be exploited in order to derive the receiver operating characteristic (ROC) of the ANMF-RTE detector.

\label{sec:background}
\subsection{Optimal design of the ANMF-RTE detector}
As explained above, in order to allow for an optimal design of the ANMF-RTE detector, one needs to characterize the distribution of $\widehat{T}_N^{\rm RTE}(\rho)$ under hypotheses $H_0$ and $H_1$. Using Proposition \ref{th:second_order}, we know that the statistic $\widehat{T}_N^{\rm RTE}(\rho)$ which cannot be handled directly, has the same fluctuations as   $\widetilde{T}_N^{\rm RTE}(\rho)$ obtained by replacing $\hat{\bf C}_N(\rho)$ by $\hat{\bf S}_N(\rho)$. That is:
$$
\widetilde{T}_N^{\rm RTE}(\rho)=\frac{\left|{\bf y}^*\hat{\bf S}_N^{-1}(\rho){\bf p}\right|}{\sqrt{{\bf p}^*\hat{\bf S}_N^{-1}(\rho){\bf p}}\sqrt{{\bf y}^*\hat{\bf S}_N^{-1}(\rho){\bf y}}}
$$
where $\hat{\bf S}_N(\rho)$ is given by \eqref{eq:hatSN}. 

Let $\underline{\rho}=\rho\left({\rho+\frac{1}{\gamma_N(\rho)}\frac{1-\rho}{1-(1-\rho)c}}\right)^{-1}$. Then, $\hat{\bf S}_N(\rho)={\rho}{\underline{\rho}^{-1}}{\widehat{\bf R}}_N(\underline{\rho})$, where, with a slight abuse of notation, we denote by ${\widehat{\bf R}}_N(\rho)$ the matrix $(1-\rho)\frac{1}{n}\sum_{i=1}^n {\bf z}_i{\bf z}_i^* +\rho {\bf I}_N$. Since $\widetilde{T}_N^{\rm RTE}(\rho)$ remains unchanged after scaling of $\hat{\bf S}_N(\rho)$ and ${\bf y}$, we also have:
$$
\widetilde{T}_N^{\rm RTE}(\rho)=\frac{\left|\frac{1}{\sqrt{\tau}}{\bf y}^*{\widehat{\bf R}}_N^{-1}(\underline{\rho}){\bf p}\right|}{\sqrt{{\bf p}^*{\widehat{\bf R}}_N^{-1}(\underline{\rho}){\bf p}}\sqrt{\frac{1}{\tau}{\bf y}^*{\widehat{\bf R}}_N^{-1}(\underline{\rho}){\bf y}}}
$$
where $\tau=1$ under $H_0$. 
It turns out that, conditionally to $\tau$, the fluctuations of the robust statistic $\widehat{T}_N^{\rm RTE}(\rho)$ under  $H_0$ or $H_1$ are the same as those obtained in Theorem \ref{th:false_alarm} and Theorem~\ref{th:detection} once $a$ is replaced by $\frac{a}{\sqrt{\tau}}$ and ${\rho}$ by $\underline{\rho}$\footnote{Note that vector ${\bf y}$ can be assumed to be Gaussian without impacting the asymptotic distributions of $\sqrt{N}\widehat{T}_N^{RTE}$ under $H_0$ and $H_1$. }. As a consequence, we have the following results:
\begin{theorem}[False alarm probability, \cite{couillet-kammoun-14}]
As $N,n\to \infty$ with $c_N\to c\in(0,\infty)$, 
$$
\sup_{\rho\in\mathcal{R}_\kappa^{\rm RTE}} \left|\mathbb{P}\left[\widehat{T}^{\rm RTE}_N(\rho) > \frac{r}{\sqrt{N}}|H_0\right]-e^{-\frac{r^2}{2\sigma_{N,{\rm RTE}}^2({\rho})}}\right|\to 0,
$$
where $\rho\mapsto \underline{\rho}$ is the aforementioned mapping and 
\begin{align*}
\sigma_{N,{\rm RTE}}^2({\rho})&\triangleq \frac{1}{2}\frac{{\bf p}^*{\bf C}_N{\bf Q}_N^2(\underline{\rho}){\bf p}}{{\bf p}^*{\bf Q}_N(\underline{\rho}){\bf p}\frac{1}{N}\tr {\bf C}_N{\bf Q}_N(\underline{\rho})}\\
&\times \frac{1}{\left(1-c(1-\underline{\rho})^2m(-\underline{\rho})^2{\frac{1}{N}\tr {\bf C}_N^2{\bf Q}_N^2(\underline{\rho})}\right)}
\end{align*}
with ${\bf Q}_N(\underline{\rho})\triangleq \left({\bf I}_N+(1-\underline{\rho})m(-\underline{\rho}){{\bf C}_N}\right)^{-1}$.
\label{th:false_alarm_1}
\end{theorem}
\begin{theorem}[Detection probability]
As $N,n\to \infty$ with $c_N\to c\in(0,\infty)$, 
\begin{align*}
&\sup_{\rho\in\mathcal{R}_\kappa^{\rm RTE}} \left|\mathbb{P}\left[\widehat{T}_N^{\rm RTE}(\rho) > \frac{r}{\sqrt{N}}|H_1\right]\right.\\
&\left.-\mathbb{E}\left[Q_1\left(g_{\rm RTE}({\bf p}),\frac{r}{\sigma_{N,{\rm RTE}}({\rho})}\right)\right]\right|\to 0,
\end{align*}
where the expectation is taken over the distribution of $\tau$, $\sigma_{N,{\rm RTE}}({\rho})$ has the same expression as in Theorem \ref{th:false_alarm_1} and 
\begin{align*}
g_{\rm RTE}({\bf p})&=\frac{\sqrt{1-c(1-\underline{\rho})^2m(-\underline{\rho})\frac{1}{N}\tr{\bf C}_N^2{\bf Q}_N^2(\underline{\rho})}}{\sqrt{{\bf p}^*{\bf C}_N{\bf Q}_N^2(\underline{\rho}){\bf p}}} \\
&\times \sqrt{\frac{2}{N\tau}} a \left|{\bf p}^*{\bf Q}_N(\underline{\rho}){\bf p}\right|.
\end{align*} and $Q_1$ is the Marcum Q-function.
\label{th:detection_1}
\end{theorem}
\begin{proof}
Since the fluctuations of the robust statistic $\widehat{T}_N^{RTE}(\rho)$ is the same as that of $\widehat{T}_N^{RSCM}(\underline{\rho})$ when $a$ is replaced by $\frac{a}{\sqrt{\tau}}$, we have for any fixed $\tau$,
\begin{align*}
&\sup_{\rho\in\mathcal{R}_{\kappa}^{RTE}} \left|\mathbb{P}\left[\widehat{T}_N^{RTE}(\rho) > \frac{r}{\sqrt{N}}|H_1,\tau\right]\right.\\
&\left.-Q_1\left(g_{RTE}({\bf p}),\frac{r}{\sigma_{N,RTE}({\rho})}\right)\right|\asto 0.
\end{align*}
The result thus follows by noticing the following inequality
\begin{align*}
&\sup_{\rho\in\mathcal{R}_\kappa^{RTE}} \left|\mathbb{P}\left[\widehat{T}_N^{RTE}(\rho) > \frac{r}{\sqrt{N}}|H_1\right]\right.\\
&\left.-\mathbb{E}\left[Q_1\left(g_{RTE}({\bf p}),\frac{r}{\sigma_{N,RTE}({\rho})}\right)\right]\right|\\
&\leq \mathbb{E}\sup_{\rho\in\mathcal{R}_\kappa^{RTE}} \left|\mathbb{P}\left[\widehat{T}_N^{RTE}(\rho) > \frac{r}{\sqrt{N}}|H_1,\tau\right]\right.\\
&-\left.Q_1\left(g_{RTE}({\bf p}),\frac{r}{\sigma_{N,RTE}({\rho})}\right)\right|
\end{align*}
and resorting to the dominated convergence theorem.
\end{proof}
Similar to the Gaussian case, we need to build  consistent estimates for $\sigma_{N,{\rm RTE}}^2({\rho})$ and $f_{\rm RTE}(\rho)$ given by:
$$
f_{\rm RTE}(\rho)=\frac{\tau}{2a^2}g_{\rm RTE}^2({\bf p})
$$
A consistent estimate for $\sigma_{N,{\rm RTE}}^2({\rho})$ was provided in \cite{couillet-kammoun-14}:
\begin{proposition}[Proposition 1 in \cite{couillet-kammoun-14}]
For $\rho\in\left(\max(\left\{0,1-c_N^{-1}\right\},1\right)$.
Define,
\begin{align*}
\hat{\sigma}_{N,{\rm RTE}}^{2}({\rho})&=\frac{1}{2}\frac{1-\rho\frac{{\bf p}^*\hat{\bf C}_N^{-2}(\rho){\bf p}}{{\bf p}^*\hat{\bf C}_N^{-1}(\rho){\bf p}}}{\left(1-c_N+c_N\rho\right)\left(1-\rho\right)} \\
\end{align*}
and let $\hat{\sigma}_{N,{\rm RTE}}^2(1)\triangleq \lim_{\rho\uparrow 1}\hat{\sigma}_N^2(\rho)$. Then, we have:
$$
\sup_{\rho\in\mathcal{R}_\kappa^{\rm RTE}}\left|\sigma_{N,{\rm RTE}}^2(\rho)-\hat{\sigma}_{N,{\rm RTE}}^2(\rho)\right|\asto 0.
$$
\end{proposition}
Similar to the Gaussian clutter case, acquiring a consistent estimate for $f_{\rm RTE}(\rho)$ is mandatory for our design. We thus prove the following Proposition:
\begin{proposition}
For $\rho\in\left(\max\left\{0,1-c_N^{-1}\right\},1\right)$, let 
\begin{align*}
\hat{f}_{\rm RTE}({\rho})&= \left({\bf p}^*\hat{\bf C}_N^{-1}(\rho){\bf p}\right)^2\left(\frac{1}{N}\tr \hat{\bf C}_N(\rho)-\rho\right)\\
&\times \frac{\left(1-c_N+c_N\rho\right)^2}{{\bf p}^*\hat{\bf C}_N^{-1}(\rho){\bf p}-\rho{\bf p}^*\hat{\bf C}_N^{-2}(\rho){\bf p}}
\end{align*}
and $\hat{f}_{{\rm RTE}}\triangleq \lim_{\rho\uparrow 1} \hat{f}_{\rm RTE}(\rho)$.
Then, we have:
$$
\sup_{\rho\in\mathcal{R}_\kappa^{\rm RTE}}\left|\hat{f}_{\rm RTE}({\rho})-f_{\rm RTE}({\rho})\right| \asto 0.
$$
\label{prop:f_rho_robust}
\end{proposition}
\begin{proof}
The proof follows by first replacing $\widehat{\bf R}_N^{-1}(\rho)$ by $\widehat{\bf R}_N^{-1}(\underline{\rho})$ and $\rho$ by $\underline{\rho}$  in the results of Proposition \ref{prop:frho}  and using the convergences \cite{couillet-kammoun-14}:
$$
\sup_{\rho\in\mathcal{R}_{\kappa}^{RTE}}\left\|\frac{\hat{\bf C}_N(\rho)}{\frac{1}{N}\tr \hat{\bf C}_N(\rho)} -\widehat{\bf R}_N(\underline{\rho})\right\|\asto 0.
$$
$$
\sup_{\rho\in\mathcal{R}_{\kappa}^{RTE}}\left\|\underline{\rho}{\hat{\bf C}_N(\rho)}-\rho\widehat{\bf R}_N(\underline{\rho})\right\|\asto 0.
$$
and
$$
\left|\frac{\rho}{\underline{\rho}}-\frac{1}{N}\tr\hat{\bf C}_N(\rho)\right|\asto 0.
$$
\end{proof}
Since the results in Proposition \ref{prop:f_rho_robust} and Theorem \ref{th:detection_1} are uniform in $\rho$, we have the following corollary:
\begin{corollary}
 Let $\hat{f}_{\rm RTE}(\rho)$ be defined as in Proposition \ref{prop:frho}. Define $\hat{\rho}_N^*$ as any value satisfying:
$$
\hat{\rho}_N^*\in\arg\max_{\rho\in\mathcal{R}_{\kappa}^{\rm RTE}}\left\{\hat{f}_{\rm RTE}(\rho)\right\}.
$$
Then, for every $r>0$,
\begin{align*}
&\mathbb{P}\left(\sqrt{N}T_N(\hat{\rho}_N^*) >r |H_1\right) \\
&-\max_{\rho\in\mathcal{R}_{\kappa}^{\rm RTE}}\left\{\mathbb{P}\left(\sqrt{N}T_N({\rho})>r\right)|H_1\right\}\asto0.
\end{align*}
\end{corollary}
Using the same reasoning as the one followed in the Gaussian clutter case, we propose the following design strategy:
\begin{itemize}
\item First, set the regularization parameter to one of the values maximizing $\hat{f}_{RTE}(\rho)$:
$$
\hat{\rho}_N^*\in\arg\max_{\rho\in\mathcal{R}_\kappa^{RTE}}\left\{\hat{f}_{RTE}(\rho)\right\};
$$
\item Second, set the threshold to $\hat{r}$
$$
\hat{r}=\hat{\sigma}_{N,RTE}(\hat{\rho}_N^*)\sqrt{-2\log\eta}
$$
where $\eta$ is the required false alarm probability.
\end{itemize}
\section{Numerical results}
\subsection{Gaussian clutter}
In a first experiment, we consider the scenario where the clutter is Gaussian with covariance matrix ${\bf C}_N$ of Toeplitz form:
\begin{equation}
\left[{\bf C}_N\right]_{i,j}=\left\{\begin{array}{ll}b^{j-i} &\hspace{0.1cm} i\leq j\\
\left(b^{i-j}\right)^*  &\hspace{0.1cm}i>j
\end{array}
, \hspace{0.9cm}|b|\in\left]0,1\right[,
\right.
\label{eq:CN}
\end{equation}
where we set $b=0.96\jmath$ $N=30$ and $n=60$. The steering vector ${\bf p}$  is given by \begin{equation}{\bf p}={\bf a}(\theta)
\label{eq:steering_vector}
\end{equation} where $\theta\mapsto \left[{\bf a}(\theta)\right]_k=e^{-\jmath \pi k\sin(\theta)}$. In this experiment, $\theta$ is set to $20^{o}$. 
For each Monte Carlo  trial, the simulated data consists of  ${\bf y}=\alpha {\bf p}+{\bf x}$  and the secondary data ${\bf y}_1,\cdots,{\bf y}_n$ which are used to estimate $\hat{\rho}_N^*$ and to compute $\hat{\bf R}_N(\hat{\rho}_N^*)$.  In particular, the shrinkage parameter and the threshold value  are determined using \eqref{eq:rho_N} and \eqref{eq:threshold}. We have observed from the considered numerical results that $\hat{f}_{\rm SCM}(\rho)$ is unimodal and thus the maximum can be obtained using efficient line search methods. For comparison, we consider two other designs: the first one is based on the  regularization parameter derived  in the work of Chen et al. \cite[Equation (19)]{chen-10} (we denote by $\hat{\rho}_{\rm chen}$ the corresponding coefficient) while the second one corresponds to the non-regularized ANMF ($\hat{\rho}=0$). . In order to satisfy the required false alarm probability, we assume for the first design that the threshold $\hat{r}_{\rm chen}$ is given by: $\hat{r}_{\rm chen}=\hat{\sigma}_{N,SCM}(\hat{\rho}_{\rm chen})\sqrt{-2\log\eta}$. For the non-regularized ANMF to satisfy the false alarm probability, the threshold value is set based on Equation 11 in \cite{Pascal04-1}.
Figure \ref{fig:gaussian_reduced} represents the ROC curves of both designs for different values of $a={\alpha}{\sqrt{N}}$, namely $a=0.1,0.25,0.5$, along with that of the theoretical performances of our design.
 We note that for all SNR ranges, the proposed algorithm outperforms the design based on the regularization parameter $\hat{\rho}_{\rm chen}$ and the gain becomes higher as $a$ increases. It also outperforms the non-regularized ANMF detector.  
Moreover,  the performances of the proposed design correspond with a good accuracy to what is expected by our theoretical results.

In order to highlight the gain of the proposed design over the most interesting range of low false alarm probabilities, we represent in Fig. \ref{fig:gaussian_low} the obtained ROC curves when $a=0.8$ and the false alarm probability spanning the interval $[0.001,0.05]$. 

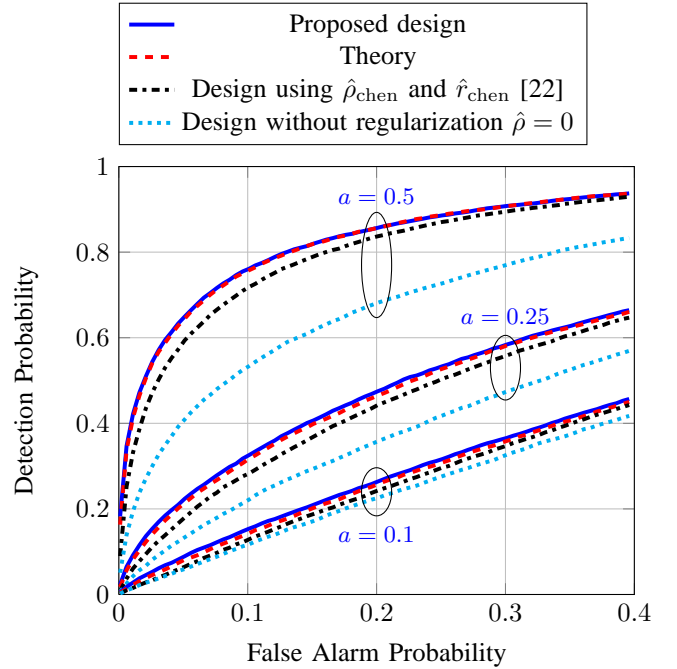
\begin{figure}
 \begin{center}
   \begin{tikzpicture}[scale=1,font=\normalsize]
\tikzset{dashdot/.style={dash pattern=on 2pt off 2pt on 6pt off 2pt}}
\tikzstyle{dotted}=[dash pattern=on \pgflinewidth off 2pt]
     \begin{axis}[
      xmin=0,
      ymin=0,
      xmax=0.4,
      ymax=1,
      grid=major,
      legend style={ at={(0,0)},
      anchor= south west,
      at={(axis cs:0,1.05)}},
      xlabel={False Alarm Probability},
      ylabel={Detection Probability},
      ]
\addplot[color=blue,line width=1.5pt,mark size=1.5pt] coordinates
{
(0.001000,0.182300)(0.006000,0.347100)(0.011000,0.421700)(0.016000,0.472300)(0.021000,0.510000)(0.026000,0.544000)(0.031000,0.571900)(0.036000,0.595300)(0.041000,0.613600)(0.046000,0.633100)(0.051000,0.649100)(0.056000,0.664400)(0.061000,0.679300)(0.066000,0.693700)(0.071000,0.704600)(0.076000,0.716200)(0.081000,0.726800)(0.086000,0.735600)(0.091000,0.746300)(0.096000,0.754800)(0.101000,0.761100)(0.106000,0.769100)(0.111000,0.774800)(0.116000,0.781500)(0.121000,0.788800)(0.126000,0.795500)(0.131000,0.800900)(0.136000,0.806600)(0.141000,0.810500)(0.146000,0.814200)(0.151000,0.819800)(0.156000,0.824500)(0.161000,0.829300)(0.166000,0.832800)(0.171000,0.837000)(0.176000,0.840400)(0.181000,0.842800)(0.186000,0.846800)(0.191000,0.850000)(0.196000,0.853400)(0.201000,0.857000)(0.206000,0.860900)(0.211000,0.863500)(0.216000,0.866900)(0.221000,0.869500)(0.226000,0.872400)(0.231000,0.875300)(0.236000,0.878100)(0.241000,0.881200)(0.246000,0.883300)(0.251000,0.885900)(0.256000,0.887600)(0.261000,0.890400)(0.266000,0.892200)(0.271000,0.894400)(0.276000,0.897300)(0.281000,0.899900)(0.286000,0.902500)(0.291000,0.904100)(0.296000,0.905900)(0.301000,0.907900)(0.306000,0.909300)(0.311000,0.910700)(0.316000,0.912300)(0.321000,0.913500)(0.326000,0.915400)(0.331000,0.917300)(0.336000,0.919800)(0.341000,0.921700)(0.346000,0.923200)(0.351000,0.924500)(0.356000,0.926300)(0.361000,0.928000)(0.366000,0.929300)(0.371000,0.930800)(0.376000,0.932000)(0.381000,0.933300)(0.386000,0.935000)(0.391000,0.935900)(0.396000,0.937300)	
};
\addlegendentry{Proposed design};
\addplot[color=red,line width=1.5pt,mark size=1.5pt,dashed] coordinates
{
(0.001000,0.163200)(0.006000,0.337300)(0.011000,0.413000)(0.016000,0.464900)(0.021000,0.505400)(0.026000,0.541000)(0.031000,0.567600)(0.036000,0.588400)(0.041000,0.609300)(0.046000,0.628100)(0.051000,0.644100)(0.056000,0.661100)(0.061000,0.674500)(0.066000,0.687400)(0.071000,0.700700)(0.076000,0.712200)(0.081000,0.722100)(0.086000,0.732500)(0.091000,0.743000)(0.096000,0.750300)(0.101000,0.757600)(0.106000,0.764600)(0.111000,0.773100)(0.116000,0.779800)(0.121000,0.786200)(0.126000,0.791800)(0.131000,0.798000)(0.136000,0.803500)(0.141000,0.808900)(0.146000,0.814500)(0.151000,0.819300)(0.156000,0.822600)(0.161000,0.827100)(0.166000,0.832200)(0.171000,0.836000)(0.176000,0.838800)(0.181000,0.843000)(0.186000,0.847600)(0.191000,0.850200)(0.196000,0.853400)(0.201000,0.856800)(0.206000,0.860100)(0.211000,0.863900)(0.216000,0.867600)(0.221000,0.870800)(0.226000,0.874300)(0.231000,0.877400)(0.236000,0.879500)(0.241000,0.882800)(0.246000,0.884900)(0.251000,0.886800)(0.256000,0.889800)(0.261000,0.891500)(0.266000,0.894000)(0.271000,0.896200)(0.276000,0.898200)(0.281000,0.899300)(0.286000,0.900900)(0.291000,0.903100)(0.296000,0.905400)(0.301000,0.907400)(0.306000,0.909300)(0.311000,0.911500)(0.316000,0.913200)(0.321000,0.914800)(0.326000,0.916300)(0.331000,0.918400)(0.336000,0.920300)(0.341000,0.921800)(0.346000,0.923500)(0.351000,0.925400)(0.356000,0.927100)(0.361000,0.928500)(0.366000,0.930200)(0.371000,0.931500)(0.376000,0.932400)(0.381000,0.933300)(0.386000,0.934900)(0.391000,0.935500)(0.396000,0.936700)	
};
\addlegendentry{Theory};
\addplot[color=black,line width=1.5pt,dashdotted] coordinates
{
(0.001000,0.089700)(0.006000,0.252400)(0.011000,0.331900)(0.016000,0.386700)(0.021000,0.435200)(0.026000,0.472100)(0.031000,0.502900)(0.036000,0.530500)(0.041000,0.552900)(0.046000,0.572400)(0.051000,0.593500)(0.056000,0.609500)(0.061000,0.623700)(0.066000,0.638100)(0.071000,0.653400)(0.076000,0.665300)(0.081000,0.676900)(0.086000,0.687800)(0.091000,0.698300)(0.096000,0.709400)(0.101000,0.718600)(0.106000,0.726100)(0.111000,0.735700)(0.116000,0.744000)(0.121000,0.751700)(0.126000,0.760300)(0.131000,0.765700)(0.136000,0.771500)(0.141000,0.777900)(0.146000,0.783800)(0.151000,0.789700)(0.156000,0.795500)(0.161000,0.800300)(0.166000,0.805000)(0.171000,0.808300)(0.176000,0.814300)(0.181000,0.819300)(0.186000,0.823800)(0.191000,0.828500)(0.196000,0.833300)(0.201000,0.837100)(0.206000,0.840400)(0.211000,0.844200)(0.216000,0.847200)(0.221000,0.850100)(0.226000,0.854100)(0.231000,0.857500)(0.236000,0.861000)(0.241000,0.864300)(0.246000,0.867400)(0.251000,0.870600)(0.256000,0.872600)(0.261000,0.875500)(0.266000,0.878500)(0.271000,0.881100)(0.276000,0.884400)(0.281000,0.886600)(0.286000,0.888500)(0.291000,0.890700)(0.296000,0.892800)(0.301000,0.894900)(0.306000,0.897400)(0.311000,0.899300)(0.316000,0.900900)(0.321000,0.903300)(0.326000,0.904900)(0.331000,0.906400)(0.336000,0.908400)(0.341000,0.909900)(0.346000,0.912800)(0.351000,0.914300)(0.356000,0.916300)(0.361000,0.918100)(0.366000,0.919800)(0.371000,0.921100)(0.376000,0.922700)(0.381000,0.924300)(0.386000,0.926300)(0.391000,0.927900)(0.396000,0.929400)	
};
\addlegendentry{Design using $\hat{\rho}_{\rm chen}$ and $\hat{r}_{\rm chen}$ \cite{chen-10}};
\addplot[color=cyan,line width=1.5pt,dotted] coordinates
{
(0.001000,0.050800)(0.006000,0.139000)(0.011000,0.194600)(0.016000,0.238100)(0.021000,0.271200)(0.026000,0.300200)(0.031000,0.327600)(0.036000,0.349500)(0.041000,0.370300)(0.046000,0.391900)(0.051000,0.408700)(0.056000,0.425600)(0.061000,0.440400)(0.066000,0.455100)(0.071000,0.466500)(0.076000,0.479800)(0.081000,0.491100)(0.086000,0.501400)(0.091000,0.513900)(0.096000,0.524500)(0.101000,0.533600)(0.106000,0.543800)(0.111000,0.553500)(0.116000,0.562600)(0.121000,0.571600)(0.126000,0.581300)(0.131000,0.590500)(0.136000,0.596600)(0.141000,0.603300)(0.146000,0.609100)(0.151000,0.617100)(0.156000,0.624700)(0.161000,0.631300)(0.166000,0.638600)(0.171000,0.643700)(0.176000,0.650700)(0.181000,0.658700)(0.186000,0.664200)(0.191000,0.670500)(0.196000,0.676200)(0.201000,0.681900)(0.206000,0.686700)(0.211000,0.691700)(0.216000,0.695600)(0.221000,0.700500)(0.226000,0.706100)(0.231000,0.710200)(0.236000,0.714800)(0.241000,0.719100)(0.246000,0.723100)(0.251000,0.727100)(0.256000,0.731700)(0.261000,0.736700)(0.266000,0.740900)(0.271000,0.745500)(0.276000,0.750100)(0.281000,0.753800)(0.286000,0.758300)(0.291000,0.762300)(0.296000,0.766100)(0.301000,0.770300)(0.306000,0.774700)(0.311000,0.778600)(0.316000,0.782800)(0.321000,0.786000)(0.326000,0.789300)(0.331000,0.793700)(0.336000,0.797500)(0.341000,0.802000)(0.346000,0.805800)(0.351000,0.808500)(0.356000,0.811900)(0.361000,0.814200)(0.366000,0.817100)(0.371000,0.819800)(0.376000,0.822900)(0.381000,0.825300)(0.386000,0.828000)(0.391000,0.830800)(0.396000,0.833200)	
};
\draw (axis cs:0.2,0.77) ellipse [x radius=0.2cm,y radius=0.7cm];
\node[blue,below] at (axis cs:0.2,0.97){{\small $a=0.5$}};
\addplot[color=blue,line width=1.5pt,mark size=1.5pt] coordinates
{
(0.001000,0.023400)(0.006000,0.064800)(0.011000,0.092000)(0.016000,0.115300)(0.021000,0.136400)(0.026000,0.153700)(0.031000,0.169900)(0.036000,0.184500)(0.041000,0.197400)(0.046000,0.209500)(0.051000,0.223500)(0.056000,0.235500)(0.061000,0.246900)(0.066000,0.256500)(0.071000,0.267500)(0.076000,0.278300)(0.081000,0.288300)(0.086000,0.298300)(0.091000,0.305800)(0.096000,0.317400)(0.101000,0.325700)(0.106000,0.334500)(0.111000,0.343200)(0.116000,0.351800)(0.121000,0.359900)(0.126000,0.369700)(0.131000,0.377100)(0.136000,0.385100)(0.141000,0.392700)(0.146000,0.400800)(0.151000,0.406700)(0.156000,0.413800)(0.161000,0.421200)(0.166000,0.429300)(0.171000,0.435500)(0.176000,0.442200)(0.181000,0.448200)(0.186000,0.455800)(0.191000,0.462600)(0.196000,0.468400)(0.201000,0.475000)(0.206000,0.480800)(0.211000,0.488200)(0.216000,0.494800)(0.221000,0.502500)(0.226000,0.508200)(0.231000,0.511700)(0.236000,0.517100)(0.241000,0.523600)(0.246000,0.528100)(0.251000,0.533400)(0.256000,0.538000)(0.261000,0.544200)(0.266000,0.550500)(0.271000,0.554200)(0.276000,0.559600)(0.281000,0.564200)(0.286000,0.569200)(0.291000,0.575200)(0.296000,0.578900)(0.301000,0.584100)(0.306000,0.588600)(0.311000,0.594200)(0.316000,0.599500)(0.321000,0.604500)(0.326000,0.608400)(0.331000,0.612200)(0.336000,0.616800)(0.341000,0.620600)(0.346000,0.624200)(0.351000,0.627900)(0.356000,0.632700)(0.361000,0.636900)(0.366000,0.641100)(0.371000,0.644900)(0.376000,0.649100)(0.381000,0.653300)(0.386000,0.657100)(0.391000,0.661300)(0.396000,0.664400)	
};
\addlegendentry{Design without regularization $\hat{\rho}=0$}
\addplot[color=red,line width=1.5pt,mark size=1.5pt,dashed] coordinates
{
	(0.001000,0.017500)(0.006000,0.056600)(0.011000,0.082300)(0.016000,0.103100)(0.021000,0.122900)(0.026000,0.141100)(0.031000,0.157200)(0.036000,0.173300)(0.041000,0.186700)(0.046000,0.201500)(0.051000,0.213700)(0.056000,0.225500)(0.061000,0.237900)(0.066000,0.247600)(0.071000,0.258100)(0.076000,0.267300)(0.081000,0.277100)(0.086000,0.287600)(0.091000,0.297400)(0.096000,0.306800)(0.101000,0.317200)(0.106000,0.325400)(0.111000,0.334800)(0.116000,0.343300)(0.121000,0.352300)(0.126000,0.361200)(0.131000,0.369300)(0.136000,0.376500)(0.141000,0.383900)(0.146000,0.390500)(0.151000,0.398700)(0.156000,0.404700)(0.161000,0.410800)(0.166000,0.417600)(0.171000,0.424000)(0.176000,0.431900)(0.181000,0.438300)(0.186000,0.445100)(0.191000,0.452800)(0.196000,0.458500)(0.201000,0.465100)(0.206000,0.470800)(0.211000,0.477200)(0.216000,0.484500)(0.221000,0.491300)(0.226000,0.498000)(0.231000,0.503500)(0.236000,0.508900)(0.241000,0.515600)(0.246000,0.522200)(0.251000,0.527500)(0.256000,0.532500)(0.261000,0.537800)(0.266000,0.543100)(0.271000,0.548800)(0.276000,0.553900)(0.281000,0.559600)(0.286000,0.565600)(0.291000,0.571500)(0.296000,0.576100)(0.301000,0.581700)(0.306000,0.586000)(0.311000,0.591200)(0.316000,0.595800)(0.321000,0.600200)(0.326000,0.604700)(0.331000,0.609300)(0.336000,0.613400)(0.341000,0.617400)(0.346000,0.621100)(0.351000,0.624600)(0.356000,0.628600)(0.361000,0.632400)(0.366000,0.636700)(0.371000,0.640800)(0.376000,0.644900)(0.381000,0.648600)(0.386000,0.652600)(0.391000,0.655800)(0.396000,0.660600)
};
\addplot[color=black,line width=1.5pt,mark size=1.5pt,dashdotted] coordinates
{
(0.001000,0.008300)(0.006000,0.037200)(0.011000,0.060800)(0.016000,0.078100)(0.021000,0.094600)(0.026000,0.111400)(0.031000,0.125900)(0.036000,0.138800)(0.041000,0.152500)(0.046000,0.166100)(0.051000,0.179600)(0.056000,0.192400)(0.061000,0.205600)(0.066000,0.216700)(0.071000,0.229200)(0.076000,0.237800)(0.081000,0.250700)(0.086000,0.258200)(0.091000,0.267000)(0.096000,0.276000)(0.101000,0.282700)(0.106000,0.291700)(0.111000,0.301100)(0.116000,0.309900)(0.121000,0.317900)(0.126000,0.327800)(0.131000,0.335700)(0.136000,0.344200)(0.141000,0.351800)(0.146000,0.360400)(0.151000,0.368700)(0.156000,0.376900)(0.161000,0.385100)(0.166000,0.393500)(0.171000,0.399700)(0.176000,0.406500)(0.181000,0.414000)(0.186000,0.421000)(0.191000,0.427900)(0.196000,0.434800)(0.201000,0.442900)(0.206000,0.448900)(0.211000,0.454900)(0.216000,0.460700)(0.221000,0.466900)(0.226000,0.473500)(0.231000,0.479000)(0.236000,0.485800)(0.241000,0.491700)(0.246000,0.497100)(0.251000,0.504000)(0.256000,0.509300)(0.261000,0.513800)(0.266000,0.519800)(0.271000,0.526000)(0.276000,0.531600)(0.281000,0.537200)(0.286000,0.542300)(0.291000,0.548400)(0.296000,0.553200)(0.301000,0.558500)(0.306000,0.564000)(0.311000,0.569800)(0.316000,0.573900)(0.321000,0.578500)(0.326000,0.583200)(0.331000,0.587300)(0.336000,0.591700)(0.341000,0.598100)(0.346000,0.603200)(0.351000,0.607100)(0.356000,0.611500)(0.361000,0.616500)(0.366000,0.621400)(0.371000,0.625400)(0.376000,0.629900)(0.381000,0.635500)(0.386000,0.639500)(0.391000,0.643700)(0.396000,0.646900)	
};
\addplot[color=cyan,line width=1.5pt,mark size=1.5pt,dotted] coordinates
{
(0.001000,0.006800)(0.006000,0.025800)(0.011000,0.042800)(0.016000,0.056800)(0.021000,0.067100)(0.026000,0.081400)(0.031000,0.092400)(0.036000,0.103200)(0.041000,0.114500)(0.046000,0.123500)(0.051000,0.133400)(0.056000,0.142900)(0.061000,0.152600)(0.066000,0.162600)(0.071000,0.170700)(0.076000,0.179200)(0.081000,0.188300)(0.086000,0.195400)(0.091000,0.204400)(0.096000,0.213200)(0.101000,0.221700)(0.106000,0.231900)(0.111000,0.240700)(0.116000,0.248700)(0.121000,0.254700)(0.126000,0.262200)(0.131000,0.269000)(0.136000,0.274900)(0.141000,0.282100)(0.146000,0.288400)(0.151000,0.294100)(0.156000,0.300000)(0.161000,0.306700)(0.166000,0.312100)(0.171000,0.319100)(0.176000,0.325100)(0.181000,0.332100)(0.186000,0.339300)(0.191000,0.345400)(0.196000,0.351500)(0.201000,0.358100)(0.206000,0.364100)(0.211000,0.370000)(0.216000,0.376500)(0.221000,0.380600)(0.226000,0.386400)(0.231000,0.391800)(0.236000,0.398300)(0.241000,0.406100)(0.246000,0.412800)(0.251000,0.418100)(0.256000,0.424400)(0.261000,0.429200)(0.266000,0.436100)(0.271000,0.439800)(0.276000,0.445900)(0.281000,0.452000)(0.286000,0.458500)(0.291000,0.462900)(0.296000,0.468100)(0.301000,0.474400)(0.306000,0.479300)(0.311000,0.484200)(0.316000,0.489200)(0.321000,0.494200)(0.326000,0.499000)(0.331000,0.503500)(0.336000,0.509900)(0.341000,0.516100)(0.346000,0.521300)(0.351000,0.525800)(0.356000,0.531700)(0.361000,0.536700)(0.366000,0.540800)(0.371000,0.545700)(0.376000,0.550500)(0.381000,0.555200)(0.386000,0.560400)(0.391000,0.564400)(0.396000,0.569300)	
};
\draw (axis cs:0.3,0.53) ellipse [x radius=0.2cm,y radius=0.43cm];
\node[blue,below] at (axis cs:0.3,0.69){{\small $a=0.25$}};
\addplot[color=blue,line width=1.5pt,mark size=1.5pt] coordinates
{
(0.001000,0.005500)(0.006000,0.017600)(0.011000,0.027900)(0.016000,0.037800)(0.021000,0.045400)(0.026000,0.053000)(0.031000,0.061400)(0.036000,0.068500)(0.041000,0.074100)(0.046000,0.081600)(0.051000,0.089200)(0.056000,0.097200)(0.061000,0.104100)(0.066000,0.109700)(0.071000,0.115300)(0.076000,0.122000)(0.081000,0.128600)(0.086000,0.134100)(0.091000,0.140900)(0.096000,0.147600)(0.101000,0.153800)(0.106000,0.159400)(0.111000,0.166200)(0.116000,0.172000)(0.121000,0.176300)(0.126000,0.182000)(0.131000,0.187700)(0.136000,0.193200)(0.141000,0.198600)(0.146000,0.203700)(0.151000,0.209500)(0.156000,0.215400)(0.161000,0.221200)(0.166000,0.226900)(0.171000,0.232600)(0.176000,0.237200)(0.181000,0.242200)(0.186000,0.248400)(0.191000,0.252900)(0.196000,0.258000)(0.201000,0.264000)(0.206000,0.270000)(0.211000,0.275900)(0.216000,0.280900)(0.221000,0.285700)(0.226000,0.290600)(0.231000,0.296400)(0.236000,0.302500)(0.241000,0.307600)(0.246000,0.311600)(0.251000,0.316400)(0.256000,0.321900)(0.261000,0.327900)(0.266000,0.332300)(0.271000,0.338100)(0.276000,0.343000)(0.281000,0.348700)(0.286000,0.353500)(0.291000,0.357300)(0.296000,0.361700)(0.301000,0.365800)(0.306000,0.371000)(0.311000,0.375200)(0.316000,0.380000)(0.321000,0.385400)(0.326000,0.389900)(0.331000,0.394900)(0.336000,0.400700)(0.341000,0.405200)(0.346000,0.410900)(0.351000,0.415500)(0.356000,0.420000)(0.361000,0.424800)(0.366000,0.429000)(0.371000,0.433100)(0.376000,0.438400)(0.381000,0.443000)(0.386000,0.446800)(0.391000,0.452600)(0.396000,0.456900)	
};
\addplot[color=red,line width=1.5pt,mark size=1.5pt,dashed] coordinates
{
(0.001000,0.004400)(0.006000,0.013200)(0.011000,0.022900)(0.016000,0.031300)(0.021000,0.038600)(0.026000,0.045600)(0.031000,0.051900)(0.036000,0.058700)(0.041000,0.065900)(0.046000,0.073800)(0.051000,0.081200)(0.056000,0.086400)(0.061000,0.091800)(0.066000,0.098500)(0.071000,0.106000)(0.076000,0.112900)(0.081000,0.118800)(0.086000,0.124500)(0.091000,0.130900)(0.096000,0.136900)(0.101000,0.143200)(0.106000,0.150500)(0.111000,0.154700)(0.116000,0.161100)(0.121000,0.167400)(0.126000,0.172500)(0.131000,0.177400)(0.136000,0.183600)(0.141000,0.190400)(0.146000,0.196000)(0.151000,0.202300)(0.156000,0.208200)(0.161000,0.213200)(0.166000,0.217700)(0.171000,0.222400)(0.176000,0.229100)(0.181000,0.235000)(0.186000,0.241200)(0.191000,0.247000)(0.196000,0.251600)(0.201000,0.257800)(0.206000,0.263800)(0.211000,0.269200)(0.216000,0.273200)(0.221000,0.279700)(0.226000,0.284100)(0.231000,0.289200)(0.236000,0.295000)(0.241000,0.298400)(0.246000,0.302800)(0.251000,0.308300)(0.256000,0.313400)(0.261000,0.319300)(0.266000,0.325100)(0.271000,0.330500)(0.276000,0.335300)(0.281000,0.340000)(0.286000,0.344100)(0.291000,0.348200)(0.296000,0.353700)(0.301000,0.358800)(0.306000,0.363400)(0.311000,0.369500)(0.316000,0.373200)(0.321000,0.379000)(0.326000,0.383800)(0.331000,0.389000)(0.336000,0.393900)(0.341000,0.399600)(0.346000,0.403400)(0.351000,0.407600)(0.356000,0.412000)(0.361000,0.417200)(0.366000,0.421900)(0.371000,0.427500)(0.376000,0.431400)(0.381000,0.436100)(0.386000,0.440700)(0.391000,0.445300)(0.396000,0.450100)	
};
\addplot[color=black,line width=1.5pt,mark size=1.5pt,dashdotted] coordinates
{
(0.001000,0.002100)(0.006000,0.009800)(0.011000,0.015400)(0.016000,0.022100)(0.021000,0.028300)(0.026000,0.033400)(0.031000,0.040000)(0.036000,0.045800)(0.041000,0.052500)(0.046000,0.058500)(0.051000,0.064800)(0.056000,0.071900)(0.061000,0.077800)(0.066000,0.085500)(0.071000,0.091100)(0.076000,0.097300)(0.081000,0.104200)(0.086000,0.109500)(0.091000,0.117200)(0.096000,0.121900)(0.101000,0.128900)(0.106000,0.135800)(0.111000,0.141100)(0.116000,0.146800)(0.121000,0.153200)(0.126000,0.159100)(0.131000,0.165700)(0.136000,0.171800)(0.141000,0.177000)(0.146000,0.183000)(0.151000,0.188000)(0.156000,0.193400)(0.161000,0.199000)(0.166000,0.203400)(0.171000,0.209700)(0.176000,0.215000)(0.181000,0.220200)(0.186000,0.224900)(0.191000,0.230500)(0.196000,0.237300)(0.201000,0.243000)(0.206000,0.248300)(0.211000,0.253000)(0.216000,0.258600)(0.221000,0.263600)(0.226000,0.269700)(0.231000,0.274300)(0.236000,0.280100)(0.241000,0.285100)(0.246000,0.290800)(0.251000,0.296500)(0.256000,0.300200)(0.261000,0.306500)(0.266000,0.312800)(0.271000,0.317800)(0.276000,0.322800)(0.281000,0.327100)(0.286000,0.331500)(0.291000,0.338200)(0.296000,0.342800)(0.301000,0.347700)(0.306000,0.354000)(0.311000,0.358500)(0.316000,0.363500)(0.321000,0.367800)(0.326000,0.372800)(0.331000,0.379000)(0.336000,0.383300)(0.341000,0.389200)(0.346000,0.394000)(0.351000,0.399400)(0.356000,0.404900)(0.361000,0.409500)(0.366000,0.415200)(0.371000,0.420000)(0.376000,0.424300)(0.381000,0.428900)(0.386000,0.433400)(0.391000,0.438000)(0.396000,0.442800)	
};
\addplot[color=cyan,line width=1.5pt,mark size=1.5pt,dotted] coordinates
{
(0.001000,0.002200)(0.006000,0.009400)(0.011000,0.015200)(0.016000,0.022000)(0.021000,0.027600)(0.026000,0.032300)(0.031000,0.037700)(0.036000,0.043000)(0.041000,0.048500)(0.046000,0.054600)(0.051000,0.059500)(0.056000,0.065300)(0.061000,0.071200)(0.066000,0.077300)(0.071000,0.082700)(0.076000,0.088300)(0.081000,0.094700)(0.086000,0.101200)(0.091000,0.106400)(0.096000,0.112500)(0.101000,0.117200)(0.106000,0.123900)(0.111000,0.129700)(0.116000,0.135200)(0.121000,0.139700)(0.126000,0.145600)(0.131000,0.152200)(0.136000,0.157300)(0.141000,0.161500)(0.146000,0.166500)(0.151000,0.171600)(0.156000,0.177500)(0.161000,0.184400)(0.166000,0.190300)(0.171000,0.195200)(0.176000,0.201200)(0.181000,0.206300)(0.186000,0.211500)(0.191000,0.216800)(0.196000,0.221600)(0.201000,0.225600)(0.206000,0.231500)(0.211000,0.236100)(0.216000,0.241500)(0.221000,0.246300)(0.226000,0.250200)(0.231000,0.254300)(0.236000,0.260300)(0.241000,0.265700)(0.246000,0.271300)(0.251000,0.276400)(0.256000,0.281500)(0.261000,0.286300)(0.266000,0.291300)(0.271000,0.295500)(0.276000,0.301100)(0.281000,0.305400)(0.286000,0.309900)(0.291000,0.314100)(0.296000,0.320700)(0.301000,0.326600)(0.306000,0.332000)(0.311000,0.336500)(0.316000,0.341200)(0.321000,0.347700)(0.326000,0.353300)(0.331000,0.357900)(0.336000,0.361500)(0.341000,0.365400)(0.346000,0.369600)(0.351000,0.374600)(0.356000,0.380300)(0.361000,0.386400)(0.366000,0.389400)(0.371000,0.395000)(0.376000,0.399400)(0.381000,0.403500)(0.386000,0.407700)(0.391000,0.413000)(0.396000,0.417300)	
};
\draw (axis cs:0.2,0.24) ellipse [x radius=0.2cm,y radius=0.32cm];
\node[blue,below] at (axis cs:0.2,0.18){{\small $a=0.1$}};
  \end{axis}
  \end{tikzpicture}
  \end{center}
\centering\caption{ROC curves of ANMF-RSCM designs for  $a=0.1,0.25,0.5$, ${\bf p}={\bf a}(\theta)$ with $\theta=20^{o}$, $N=30$, $n=60$: Gaussian setting}
\label{fig:gaussian_reduced}  
\end{figure}
\begin{figure}
  \begin{center}
   \begin{tikzpicture}[scale=1,font=\normalsize]
\tikzset{dashdot/.style={dash pattern=on 2pt off 2pt on 6pt off 2pt}}
\tikzstyle{dotted}=[dash pattern=on \pgflinewidth off 2pt]
     \begin{axis}[
      xmin=0,
      ymin=0.2,
      xmax=0.05,
      ymax=1,
      grid=major,
      legend style={ at={(0,0)},
      anchor= south west,
      at={(axis cs:0,1.05)}},
      xlabel={False Alarm Probability},
      ylabel={Detection Probability},
      ]
\addplot[color=blue,line width=1.5pt,mark size=1.5pt] coordinates
{
(0.001000,0.662400)(0.001500,0.702800)(0.002000,0.730300)(0.002500,0.753400)(0.003000,0.770000)(0.003500,0.786800)(0.004000,0.799100)(0.004500,0.810200)(0.005000,0.817600)(0.005500,0.825200)(0.006000,0.832800)(0.006500,0.839200)(0.007000,0.846000)(0.007500,0.852600)(0.008000,0.857400)(0.008500,0.861400)(0.009000,0.865400)(0.009500,0.871600)(0.010000,0.875600)(0.010500,0.878600)(0.011000,0.881900)(0.011500,0.885200)(0.012000,0.888600)(0.012500,0.891800)(0.013000,0.894700)(0.013500,0.896900)(0.014000,0.899600)(0.014500,0.902400)(0.015000,0.904600)(0.015500,0.906400)(0.016000,0.908600)(0.016500,0.910600)(0.017000,0.912300)(0.017500,0.914400)(0.018000,0.916600)(0.018500,0.918200)(0.019000,0.919900)(0.019500,0.921400)(0.020000,0.922600)(0.020500,0.923800)(0.021000,0.925000)(0.021500,0.926200)(0.022000,0.927400)(0.022500,0.929100)(0.023000,0.930000)(0.023500,0.930700)(0.024000,0.931700)(0.024500,0.932800)(0.025000,0.933800)(0.025500,0.934800)(0.026000,0.936700)(0.026500,0.937500)(0.027000,0.938800)(0.027500,0.939500)(0.028000,0.940300)(0.028500,0.941400)(0.029000,0.942300)(0.029500,0.942900)(0.030000,0.943600)(0.030500,0.944400)(0.031000,0.945200)(0.031500,0.946000)(0.032000,0.947100)(0.032500,0.947600)(0.033000,0.948000)(0.033500,0.948500)(0.034000,0.949200)(0.034500,0.950000)(0.035000,0.950900)(0.035500,0.951300)(0.036000,0.951500)(0.036500,0.952200)(0.037000,0.952500)(0.037500,0.953000)(0.038000,0.953500)(0.038500,0.954300)(0.039000,0.954800)(0.039500,0.955300)(0.040000,0.955900)(0.040500,0.956300)(0.041000,0.956900)(0.041500,0.957100)(0.042000,0.957400)(0.042500,0.957800)(0.043000,0.958500)(0.043500,0.958900)(0.044000,0.959300)(0.044500,0.959700)(0.045000,0.959900)(0.045500,0.960500)(0.046000,0.960900)(0.046500,0.961300)(0.047000,0.961800)(0.047500,0.961800)(0.048000,0.962500)(0.048500,0.962900)(0.049000,0.963100)(0.049500,0.963600)(0.050000,0.963900)	
};
\addlegendentry{Proposed design };
\addplot[color=red,line width=1.5pt,mark size=1.5pt,dashed] coordinates
{
(0.001000,0.656100)(0.001500,0.699500)(0.002000,0.729600)(0.002500,0.755600)(0.003000,0.772800)(0.003500,0.787800)(0.004000,0.798600)(0.004500,0.808600)(0.005000,0.818100)(0.005500,0.827900)(0.006000,0.836800)(0.006500,0.843300)(0.007000,0.848600)(0.007500,0.853900)(0.008000,0.859500)(0.008500,0.864000)(0.009000,0.868400)(0.009500,0.873200)(0.010000,0.877000)(0.010500,0.880500)(0.011000,0.884800)(0.011500,0.888700)(0.012000,0.891500)(0.012500,0.894500)(0.013000,0.897200)(0.013500,0.900300)(0.014000,0.902700)(0.014500,0.904800)(0.015000,0.907300)(0.015500,0.910000)(0.016000,0.911200)(0.016500,0.913200)(0.017000,0.914400)(0.017500,0.916600)(0.018000,0.918100)(0.018500,0.919500)(0.019000,0.921000)(0.019500,0.922700)(0.020000,0.924300)(0.020500,0.925900)(0.021000,0.927200)(0.021500,0.928700)(0.022000,0.929900)(0.022500,0.931500)(0.023000,0.932700)(0.023500,0.933800)(0.024000,0.934800)(0.024500,0.936300)(0.025000,0.937300)(0.025500,0.937900)(0.026000,0.939100)(0.026500,0.939300)(0.027000,0.940200)(0.027500,0.940700)(0.028000,0.941400)(0.028500,0.942500)(0.029000,0.943300)(0.029500,0.944200)(0.030000,0.944700)(0.030500,0.945300)(0.031000,0.946000)(0.031500,0.946600)(0.032000,0.947500)(0.032500,0.948500)(0.033000,0.948800)(0.033500,0.949800)(0.034000,0.950000)(0.034500,0.950500)(0.035000,0.950800)(0.035500,0.951600)(0.036000,0.952000)(0.036500,0.952800)(0.037000,0.953400)(0.037500,0.954300)(0.038000,0.954800)(0.038500,0.956100)(0.039000,0.956700)(0.039500,0.957200)(0.040000,0.957800)(0.040500,0.958200)(0.041000,0.958200)(0.041500,0.958700)(0.042000,0.959500)(0.042500,0.959900)(0.043000,0.960400)(0.043500,0.961200)(0.044000,0.961500)(0.044500,0.961900)(0.045000,0.962200)(0.045500,0.962900)(0.046000,0.963500)(0.046500,0.963600)(0.047000,0.964000)(0.047500,0.964000)(0.048000,0.964200)(0.048500,0.964800)(0.049000,0.964900)(0.049500,0.965000)(0.050000,0.965100)
};
\addlegendentry{Theory};
\addplot[color=black,line width=1.5pt,mark size=1.5pt,dashdotted] coordinates
{
(0.001000,0.502800)(0.001500,0.560100)(0.002000,0.600800)(0.002500,0.635300)(0.003000,0.659800)(0.003500,0.684700)(0.004000,0.703800)(0.004500,0.718000)(0.005000,0.732700)(0.005500,0.742900)(0.006000,0.755200)(0.006500,0.765400)(0.007000,0.776100)(0.007500,0.784500)(0.008000,0.791300)(0.008500,0.798200)(0.009000,0.805000)(0.009500,0.812700)(0.010000,0.818400)(0.010500,0.824300)(0.011000,0.829200)(0.011500,0.833500)(0.012000,0.838500)(0.012500,0.842100)(0.013000,0.845600)(0.013500,0.850200)(0.014000,0.854000)(0.014500,0.856700)(0.015000,0.860000)(0.015500,0.863600)(0.016000,0.866700)(0.016500,0.869400)(0.017000,0.871800)(0.017500,0.874900)(0.018000,0.877600)(0.018500,0.879700)(0.019000,0.882200)(0.019500,0.884500)(0.020000,0.886800)(0.020500,0.888900)(0.021000,0.890100)(0.021500,0.891900)(0.022000,0.893800)(0.022500,0.895800)(0.023000,0.898000)(0.023500,0.899600)(0.024000,0.901200)(0.024500,0.902300)(0.025000,0.904400)(0.025500,0.905900)(0.026000,0.907400)(0.026500,0.909500)(0.027000,0.911900)(0.027500,0.912900)(0.028000,0.914000)(0.028500,0.915800)(0.029000,0.916700)(0.029500,0.918200)(0.030000,0.919600)(0.030500,0.920700)(0.031000,0.922200)(0.031500,0.923000)(0.032000,0.924100)(0.032500,0.924900)(0.033000,0.926700)(0.033500,0.928400)(0.034000,0.929500)(0.034500,0.930200)(0.035000,0.930800)(0.035500,0.931900)(0.036000,0.932900)(0.036500,0.933400)(0.037000,0.934600)(0.037500,0.935800)(0.038000,0.936600)(0.038500,0.937600)(0.039000,0.938000)(0.039500,0.939000)(0.040000,0.939200)(0.040500,0.939900)(0.041000,0.940500)(0.041500,0.941300)(0.042000,0.942100)(0.042500,0.942500)(0.043000,0.943300)(0.043500,0.944200)(0.044000,0.944600)(0.044500,0.945100)(0.045000,0.945400)(0.045500,0.946400)(0.046000,0.946800)(0.046500,0.947500)(0.047000,0.947700)(0.047500,0.948200)(0.048000,0.948700)(0.048500,0.949000)(0.049000,0.950000)(0.049500,0.950500)(0.050000,0.950800)	
};
\addlegendentry{Design using $\hat{\rho}_{\rm chen}$ and $\hat{r}_{\rm chen}$ \cite{chen-10}};
\addplot[color=cyan,line width=1.5pt,mark size=1.5pt,dotted] coordinates
{
(0.001000,0.293500)(0.001500,0.336800)(0.002000,0.368300)(0.002500,0.396100)(0.003000,0.419400)(0.003500,0.438600)(0.004000,0.457000)(0.004500,0.470800)(0.005000,0.482300)(0.005500,0.493300)(0.006000,0.505600)(0.006500,0.515800)(0.007000,0.525400)(0.007500,0.534400)(0.008000,0.543100)(0.008500,0.550900)(0.009000,0.558900)(0.009500,0.567900)(0.010000,0.574500)(0.010500,0.581100)(0.011000,0.588000)(0.011500,0.593900)(0.012000,0.600500)(0.012500,0.605500)(0.013000,0.610200)(0.013500,0.614200)(0.014000,0.619000)(0.014500,0.624500)(0.015000,0.628600)(0.015500,0.633300)(0.016000,0.638500)(0.016500,0.642800)(0.017000,0.646900)(0.017500,0.651800)(0.018000,0.656300)(0.018500,0.660300)(0.019000,0.665200)(0.019500,0.669400)(0.020000,0.672700)(0.020500,0.676900)(0.021000,0.680500)(0.021500,0.684500)(0.022000,0.688000)(0.022500,0.691300)(0.023000,0.693800)(0.023500,0.696000)(0.024000,0.699100)(0.024500,0.702300)(0.025000,0.705000)(0.025500,0.708600)(0.026000,0.711900)(0.026500,0.715100)(0.027000,0.717900)(0.027500,0.719600)(0.028000,0.722100)(0.028500,0.723400)(0.029000,0.725700)(0.029500,0.728200)(0.030000,0.730100)(0.030500,0.732400)(0.031000,0.735200)(0.031500,0.736800)(0.032000,0.739000)(0.032500,0.741300)(0.033000,0.743600)(0.033500,0.746500)(0.034000,0.747900)(0.034500,0.749300)(0.035000,0.751400)(0.035500,0.753100)(0.036000,0.755700)(0.036500,0.757200)(0.037000,0.758700)(0.037500,0.760800)(0.038000,0.762600)(0.038500,0.764600)(0.039000,0.766900)(0.039500,0.769300)(0.040000,0.771500)(0.040500,0.773100)(0.041000,0.774400)(0.041500,0.775700)(0.042000,0.777700)(0.042500,0.780400)(0.043000,0.781400)(0.043500,0.783800)(0.044000,0.785500)(0.044500,0.787200)(0.045000,0.788300)(0.045500,0.790700)(0.046000,0.791900)(0.046500,0.793400)(0.047000,0.795100)(0.047500,0.796600)(0.048000,0.798700)(0.048500,0.799600)(0.049000,0.801000)(0.049500,0.802500)(0.050000,0.803200)	
};
\addlegendentry{Design without regularization $\hat{\rho}=0$}
	  \end{axis}
   \end{tikzpicture}
	   \end{center}
	   \caption{ROC curves of ANMF-RSCM designs for $a=0.9$, ${\bf p}={\bf a}(\theta)$ with $\theta=20^{o}$, $N=30$, $n=60$: Gaussian setting}
\label{fig:gaussian_low}
\end{figure}
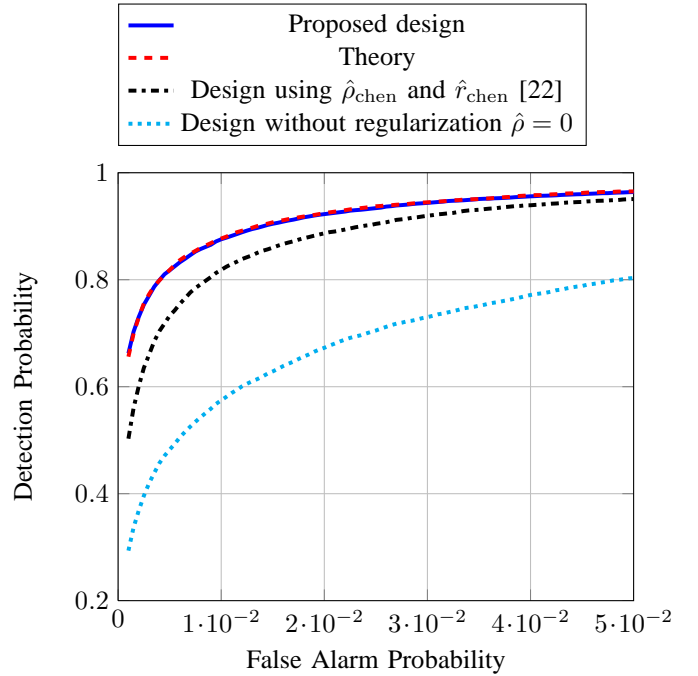

\subsection{Non-Gaussian clutter}
In a second experiment, we proceed  investigating the performance of the proposed design in the case where a non-Gaussian clutter is considered. In particular, we consider the case where the clutter is drawn from a $K$-distribution with zero mean, covariance ${\bf C}_N$, and shape $\nu=0.5$ \cite{esa-12}. The covariance matrix has the same form as in \eqref{eq:CN}  with $b=0.96\jmath$ but with $N=30$ and $n=60$.

 Similar to the Gaussian clutter case, we consider for the sake of comparison the concurrent design based on the regularization parameter derived in the work of  Olilla and Tyler in \cite[Equation (19)]{ollila-tyler}. We denote by $\hat{\rho}_{\rm ollila}$ and $\hat{r}_{\rm ollila}$ the corresponding regularization coefficient and threshold. Note that, according to our theoretical analysis, the threshold $\hat{r}_{\rm ollila}$ should be set to $\hat{r}_{\rm ollila}=\hat{\sigma}_{N,RTE}(\hat{\rho}_{\rm ollila})\sqrt{-2\log\eta}$ in order to satisfy the required false alarm probability. The results are depicted in Figure \ref{fig:comp_olilla}.
 We note that for all SNR ranges, the proposed method achieves a gain over the design based on the regularization coefficient proposed by Olilla et al. We also observe  that, similar to the first experiment,  the gain increases as $a$ grows but with a lower slope\footnote{Note that we do not compare with the zero-regularization case as in the first experiment, since, contrary to the Gaussian case, we do not have in our disposal theoretical results allowing the tuning of the threshold to the value that achieves the required false alarm probability.}.  

In a last experiment, we investigate the impacts of $a$  and the distribution shape $\nu$.  Figure \ref{fig:effect_nu} represents the detection probability with respect to $a$ when the false alarm probability is fixed to $0.05$. We note that for small values of $a$,  higher detection probabilities are achieved when the distribution of the clutter is heavy-tailed (small $\nu$), whereas  the opposite occurs for large values of $a$. In order to explain this change in behavior, we must recall that heavy-tailed clutters (small $\nu$) are characterized by a higher number of occurrences of $\tau$ in the vicinity of zero and at the same time more frequent realizations of  large values of $\tau$. If $a$ is small, the improvement in detection performances  achieved by heavy-tailed clutters is attributed to the artificial increase in SNR over realizations of small values of $\tau$. As $a$ increases, the power of the signal of interest is high enough so that the effect of realizations with large values of $\tau$ becomes  dominant. The latter, which are more frequent for small values of $\nu$, are  characterized by high levels of noises, thereby entailing a degradation of detection performances. 
\begin{figure}
 \begin{center}
   \begin{tikzpicture}[scale=1,font=\normalsize]
\tikzset{dashdot/.style={dash pattern=on 2pt off 2pt on 6pt off 2pt}}
     \begin{axis}[
 xmin=0,
 ymin=0.1,
 xmax=0.4,
 ymax=1,
 grid=major,
      legend style={ at={(0,0)},
      anchor= south west,
      at={(axis cs:0,1.05)}},
      xlabel={False Alarm Probability},
      ylabel={Detection Probability},
      ]
\addplot[color=blue,line width=1.5pt,mark size=1.5pt] coordinates
{
(0.001000,0.518900)(0.006000,0.612200)(0.011000,0.649200)(0.016000,0.672000)(0.021000,0.691700)(0.026000,0.705600)(0.031000,0.717400)(0.036000,0.729700)(0.041000,0.741100)(0.046000,0.750600)(0.051000,0.757900)(0.056000,0.764800)(0.061000,0.771000)(0.066000,0.776600)(0.071000,0.781700)(0.076000,0.786900)(0.081000,0.792500)(0.086000,0.797200)(0.091000,0.801300)(0.096000,0.805900)(0.101000,0.809900)(0.106000,0.814900)(0.111000,0.819100)(0.116000,0.822300)(0.121000,0.825200)(0.126000,0.828300)(0.131000,0.830400)(0.136000,0.833100)(0.141000,0.836000)(0.146000,0.839400)(0.151000,0.842200)(0.156000,0.845200)(0.161000,0.848400)(0.166000,0.851100)(0.171000,0.853300)(0.176000,0.855800)(0.181000,0.857500)(0.186000,0.860200)(0.191000,0.862200)(0.196000,0.865800)(0.201000,0.868300)(0.206000,0.869900)(0.211000,0.871300)(0.216000,0.873500)(0.221000,0.875000)(0.226000,0.876400)(0.231000,0.879300)(0.236000,0.881700)(0.241000,0.883400)(0.246000,0.885100)(0.251000,0.886800)(0.256000,0.888200)(0.261000,0.889800)(0.266000,0.891100)(0.271000,0.892500)(0.276000,0.894300)(0.281000,0.895500)(0.286000,0.897200)(0.291000,0.898000)(0.296000,0.899200)(0.301000,0.901100)(0.306000,0.902100)(0.311000,0.903400)(0.316000,0.905600)(0.321000,0.906700)(0.326000,0.908600)(0.331000,0.909800)(0.336000,0.911500)(0.341000,0.914000)(0.346000,0.914900)(0.351000,0.915900)(0.356000,0.916800)(0.361000,0.918100)(0.366000,0.919100)(0.371000,0.920000)(0.376000,0.921800)(0.381000,0.923300)(0.386000,0.924300)(0.391000,0.925700)(0.396000,0.926100)
};
\addlegendentry{Proposed design };
\addplot[color=black,mark size=1.5pt,dashdot,line width=1.5pt] coordinates{
(0.001000,0.518900)(0.006000,0.594400)(0.011000,0.626100)(0.016000,0.645300)(0.021000,0.661000)(0.026000,0.676100)(0.031000,0.688500)(0.036000,0.698100)(0.041000,0.706900)(0.046000,0.714700)(0.051000,0.722300)(0.056000,0.729900)(0.061000,0.736200)(0.066000,0.741100)(0.071000,0.746800)(0.076000,0.751800)(0.081000,0.756000)(0.086000,0.760900)(0.091000,0.766400)(0.096000,0.770600)(0.101000,0.774500)(0.106000,0.779000)(0.111000,0.782800)(0.116000,0.785700)(0.121000,0.789300)(0.126000,0.792500)(0.131000,0.795900)(0.136000,0.800400)(0.141000,0.802800)(0.146000,0.806500)(0.151000,0.809000)(0.156000,0.812800)(0.161000,0.814900)(0.166000,0.816900)(0.171000,0.818500)(0.176000,0.821100)(0.181000,0.823300)(0.186000,0.826200)(0.191000,0.828500)(0.196000,0.831700)(0.201000,0.833100)(0.206000,0.835400)(0.211000,0.837200)(0.216000,0.839600)(0.221000,0.842700)(0.226000,0.844400)(0.231000,0.846900)(0.236000,0.848800)(0.241000,0.850600)(0.246000,0.852600)(0.251000,0.854700)(0.256000,0.857100)(0.261000,0.858800)(0.266000,0.860600)(0.271000,0.862200)(0.276000,0.863300)(0.281000,0.865000)(0.286000,0.866500)(0.291000,0.868200)(0.296000,0.870100)(0.301000,0.871400)(0.306000,0.872500)(0.311000,0.874100)(0.316000,0.875500)(0.321000,0.876800)(0.326000,0.878300)(0.331000,0.880000)(0.336000,0.881100)(0.341000,0.882400)(0.346000,0.884000)(0.351000,0.886700)(0.356000,0.888600)(0.361000,0.889900)(0.366000,0.891500)(0.371000,0.892200)(0.376000,0.893700)(0.381000,0.894900)(0.386000,0.896400)(0.391000,0.898000)(0.396000,0.899300)
};
\addlegendentry{Design using $\hat{\rho}_{\rm olilla}$, $\hat{r}_{\rm olilla}$ \cite{ollila-tyler}};
\draw (axis cs:0.2,0.85) ellipse [x radius=0.2cm,y radius=0.4cm];
\node[blue,below] at (axis cs:0.2,0.97){{\small $a=0.5$}};
\addplot[color=blue,line width=1.5pt,mark size=1.5pt] coordinates
{
(0.001000,0.286400)(0.006000,0.357600)(0.011000,0.389200)(0.016000,0.413300)(0.021000,0.429700)(0.026000,0.445100)(0.031000,0.457400)(0.036000,0.470800)(0.041000,0.481000)(0.046000,0.492600)(0.051000,0.503200)(0.056000,0.510000)(0.061000,0.518400)(0.066000,0.526900)(0.071000,0.534100)(0.076000,0.541300)(0.081000,0.549000)(0.086000,0.556700)(0.091000,0.564700)(0.096000,0.572300)(0.101000,0.578400)(0.106000,0.585000)(0.111000,0.590900)(0.116000,0.596400)(0.121000,0.600300)(0.126000,0.605200)(0.131000,0.609700)(0.136000,0.614400)(0.141000,0.619800)(0.146000,0.623200)(0.151000,0.627400)(0.156000,0.632000)(0.161000,0.637000)(0.166000,0.640600)(0.171000,0.643800)(0.176000,0.648300)(0.181000,0.652700)(0.186000,0.657000)(0.191000,0.661000)(0.196000,0.665700)(0.201000,0.669000)(0.206000,0.672500)(0.211000,0.676100)(0.216000,0.679600)(0.221000,0.682800)(0.226000,0.686200)(0.231000,0.689500)(0.236000,0.693200)(0.241000,0.696400)(0.246000,0.699100)(0.251000,0.702800)(0.256000,0.705500)(0.261000,0.708600)(0.266000,0.711000)(0.271000,0.715000)(0.276000,0.718500)(0.281000,0.721600)(0.286000,0.724700)(0.291000,0.727500)(0.296000,0.730700)(0.301000,0.733500)(0.306000,0.735800)(0.311000,0.739200)(0.316000,0.741800)(0.321000,0.745800)(0.326000,0.749600)(0.331000,0.753200)(0.336000,0.756000)(0.341000,0.759200)(0.346000,0.760800)(0.351000,0.763200)(0.356000,0.766000)(0.361000,0.769100)(0.366000,0.772900)(0.371000,0.776400)(0.376000,0.778700)(0.381000,0.781000)(0.386000,0.783300)(0.391000,0.786200)(0.396000,0.789200)
};
\addplot[color=black,mark size=1.5pt,dashdot,line width=1.5pt] coordinates
{
(0.001000,0.296000)(0.006000,0.352800)(0.011000,0.378500)(0.016000,0.397600)(0.021000,0.410300)(0.026000,0.423400)(0.031000,0.435800)(0.036000,0.445200)(0.041000,0.454500)(0.046000,0.462600)(0.051000,0.469600)(0.056000,0.478400)(0.061000,0.485500)(0.066000,0.492700)(0.071000,0.499900)(0.076000,0.507200)(0.081000,0.512000)(0.086000,0.518200)(0.091000,0.523800)(0.096000,0.530400)(0.101000,0.536100)(0.106000,0.541300)(0.111000,0.547400)(0.116000,0.551600)(0.121000,0.556300)(0.126000,0.561400)(0.131000,0.566300)(0.136000,0.570400)(0.141000,0.575900)(0.146000,0.581600)(0.151000,0.584500)(0.156000,0.588400)(0.161000,0.592500)(0.166000,0.596300)(0.171000,0.600000)(0.176000,0.604100)(0.181000,0.608800)(0.186000,0.613800)(0.191000,0.617500)(0.196000,0.621200)(0.201000,0.625800)(0.206000,0.629400)(0.211000,0.633600)(0.216000,0.636900)(0.221000,0.640600)(0.226000,0.643400)(0.231000,0.647500)(0.236000,0.650900)(0.241000,0.653700)(0.246000,0.656600)(0.251000,0.659800)(0.256000,0.663300)(0.261000,0.667000)(0.266000,0.669800)(0.271000,0.673700)(0.276000,0.676800)(0.281000,0.680000)(0.286000,0.683300)(0.291000,0.685900)(0.296000,0.689300)(0.301000,0.692500)(0.306000,0.695500)(0.311000,0.698800)(0.316000,0.701800)(0.321000,0.704400)(0.326000,0.706300)(0.331000,0.709500)(0.336000,0.712600)(0.341000,0.715800)(0.346000,0.718400)(0.351000,0.721900)(0.356000,0.723500)(0.361000,0.726800)(0.366000,0.730200)(0.371000,0.732700)(0.376000,0.735400)(0.381000,0.738400)(0.386000,0.740600)(0.391000,0.743400)(0.396000,0.746000)
};
\draw (axis cs:0.2,0.65) ellipse [x radius=0.2cm,y radius=0.4cm];
\node[blue,below] at (axis cs:0.2,0.58){{\small $a=0.25$}};
\addplot[color=blue,line width=1.5pt,mark size=1.5pt] coordinates
{
(0.001000,0.118300)(0.006000,0.154700)(0.011000,0.175400)(0.016000,0.191300)(0.021000,0.204900)(0.026000,0.216100)(0.031000,0.224200)(0.036000,0.234500)(0.041000,0.243300)(0.046000,0.251200)(0.051000,0.259900)(0.056000,0.267800)(0.061000,0.275100)(0.066000,0.284000)(0.071000,0.291600)(0.076000,0.299400)(0.081000,0.307300)(0.086000,0.313600)(0.091000,0.321200)(0.096000,0.327500)(0.101000,0.333600)(0.106000,0.340100)(0.111000,0.348000)(0.116000,0.353300)(0.121000,0.358600)(0.126000,0.363900)(0.131000,0.370300)(0.136000,0.376500)(0.141000,0.381000)(0.146000,0.385800)(0.151000,0.390400)(0.156000,0.395900)(0.161000,0.400800)(0.166000,0.405300)(0.171000,0.410600)(0.176000,0.415600)(0.181000,0.419300)(0.186000,0.423600)(0.191000,0.427700)(0.196000,0.432500)(0.201000,0.438200)(0.206000,0.443400)(0.211000,0.448600)(0.216000,0.453200)(0.221000,0.458800)(0.226000,0.462900)(0.231000,0.466900)(0.236000,0.471300)(0.241000,0.475900)(0.246000,0.480600)(0.251000,0.485200)(0.256000,0.489500)(0.261000,0.495600)(0.266000,0.499200)(0.271000,0.504600)(0.276000,0.509500)(0.281000,0.513800)(0.286000,0.516400)(0.291000,0.520200)(0.296000,0.523700)(0.301000,0.528100)(0.306000,0.532300)(0.311000,0.536700)(0.316000,0.541700)(0.321000,0.545900)(0.326000,0.550500)(0.331000,0.555300)(0.336000,0.558500)(0.341000,0.563000)(0.346000,0.567200)(0.351000,0.571300)(0.356000,0.575800)(0.361000,0.578900)(0.366000,0.583200)(0.371000,0.587300)(0.376000,0.591300)(0.381000,0.594100)(0.386000,0.597800)(0.391000,0.601400)(0.396000,0.605800)
};
\addplot[color=black,mark size=1.5pt,dashdot,line width=1.5pt] coordinates
{
(0.001000,0.124700)(0.006000,0.159200)(0.011000,0.177800)(0.016000,0.189100)(0.021000,0.200500)(0.026000,0.212800)(0.031000,0.222800)(0.036000,0.231300)(0.041000,0.239300)(0.046000,0.249800)(0.051000,0.257100)(0.056000,0.263800)(0.061000,0.270000)(0.066000,0.275100)(0.071000,0.281000)(0.076000,0.287500)(0.081000,0.292200)(0.086000,0.298200)(0.091000,0.303400)(0.096000,0.308700)(0.101000,0.313800)(0.106000,0.319600)(0.111000,0.325700)(0.116000,0.331000)(0.121000,0.335800)(0.126000,0.341700)(0.131000,0.346700)(0.136000,0.352800)(0.141000,0.357500)(0.146000,0.362200)(0.151000,0.367100)(0.156000,0.370100)(0.161000,0.373800)(0.166000,0.379900)(0.171000,0.384500)(0.176000,0.388900)(0.181000,0.392600)(0.186000,0.397700)(0.191000,0.401700)(0.196000,0.406000)(0.201000,0.409800)(0.206000,0.412500)(0.211000,0.416900)(0.216000,0.421000)(0.221000,0.425100)(0.226000,0.429100)(0.231000,0.433200)(0.236000,0.437000)(0.241000,0.440500)(0.246000,0.444300)(0.251000,0.448100)(0.256000,0.451900)(0.261000,0.455600)(0.266000,0.458500)(0.271000,0.462900)(0.276000,0.467000)(0.281000,0.470400)(0.286000,0.474000)(0.291000,0.477700)(0.296000,0.481900)(0.301000,0.485300)(0.306000,0.488700)(0.311000,0.492400)(0.316000,0.495700)(0.321000,0.499900)(0.326000,0.502900)(0.331000,0.508100)(0.336000,0.512200)(0.341000,0.515800)(0.346000,0.520300)(0.351000,0.523400)(0.356000,0.526200)(0.361000,0.529400)(0.366000,0.532100)(0.371000,0.536100)(0.376000,0.540200)(0.381000,0.543700)(0.386000,0.546100)(0.391000,0.549800)(0.396000,0.554100)
};
\draw (axis cs:0.3,0.5) ellipse [x radius=0.2cm,y radius=0.4cm];
\node[blue,below] at (axis cs:0.3,0.42){{\small $a=0.1$}};
\end{axis}
 \end{tikzpicture}
\end{center}
\centering\caption{ROC curves of ANMF-RTE designs for $a=0.1,0.25,0.5$, ${\bf p}={\bf a}(\theta)$ with $\theta=20^{o}$, $N=30$, $n=60$: K distributed clutter setting}
\label{fig:comp_olilla}
\end{figure}
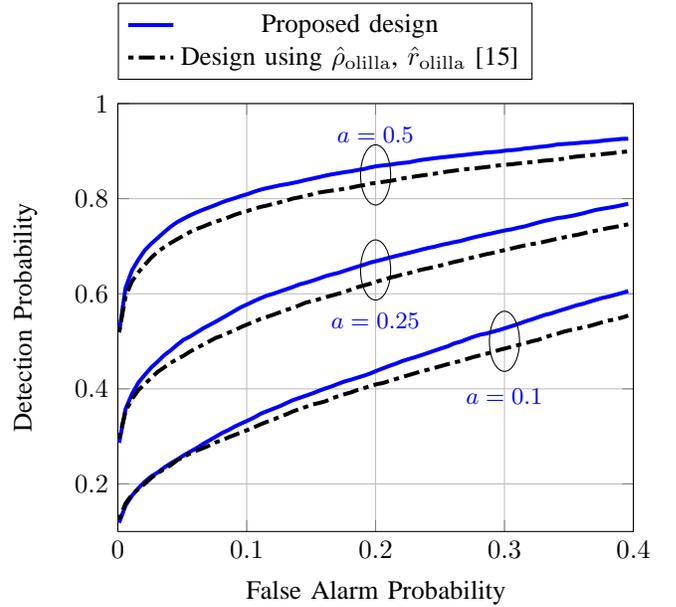

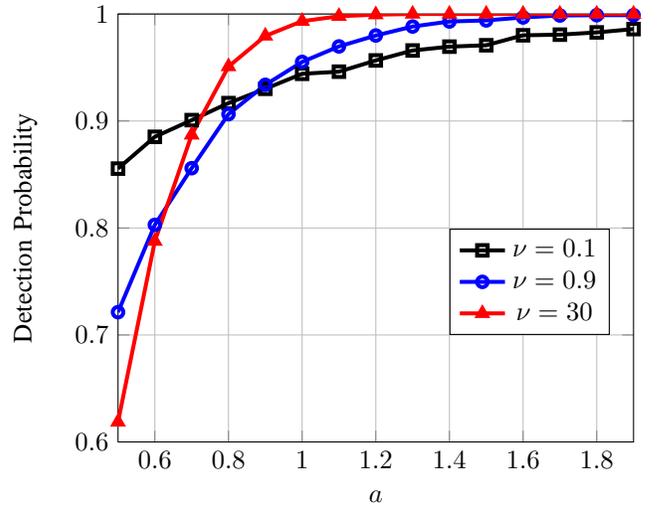
\begin{figure}
 \begin{center}
   \begin{tikzpicture}[scale=1,font=\normalsize]
\tikzset{dashdot/.style={dash pattern=on 2pt off 2pt on 6pt off 2pt}}
     \begin{axis}[
 xmin=0.5,
 ymin=0.6,
 xmax=1.9,
 ymax=1,
 grid=major,
      legend style={ at={(0,0)},
      anchor= south west,
      at={(axis cs:1.4,0.7)}},
      xlabel={$a$},
      ylabel={Detection Probability},
      ]
\addplot[color=black,line width=1.5pt,mark size=2pt,mark=square] coordinates
{
(0.500000,0.855500)(0.600000,0.885200)(0.700000,0.900800)(0.800000,0.916700)(0.900000,0.930100)(1.000000,0.944100)(1.100000,0.946100)(1.200000,0.956600)(1.300000,0.965900)(1.400000,0.969500)(1.500000,0.970800)(1.600000,0.980100)(1.700000,0.980700)(1.800000,0.982800)(1.900000,0.985900)(2.000000,0.990600)(2.100000,0.992300)(2.200000,0.990900)(2.300000,0.992400)(2.400000,0.994000)(2.500000,0.994700)(2.600000,0.995300)(2.700000,0.997300)(2.800000,0.996600)(2.900000,0.997200)(3.000000,0.997600)(3.100000,0.997800)(3.200000,0.997900)(3.300000,0.998700)(3.400000,0.998700)(3.500000,0.998600)(3.600000,0.998600)(3.700000,0.999400)(3.800000,0.999500)(3.900000,0.999700)(4.000000,0.999400)(4.100000,0.999500)(4.200000,0.999600)(4.300000,0.999700)(4.400000,0.999500)(4.500000,0.999600)(4.600000,0.999900)(4.700000,1.000000)(4.800000,0.999900)(4.900000,0.999900)(5.000000,0.999900)
};
\addlegendentry{$\nu=0.1$}

\addplot[color=blue,line width=1.5pt,mark size=2pt,mark=o] coordinates
{
(0.500000,0.721300)(0.600000,0.802900)(0.700000,0.855900)(0.800000,0.906400)(0.900000,0.933800)(1.000000,0.955200)(1.100000,0.969700)(1.200000,0.979900)(1.300000,0.988300)(1.400000,0.993000)(1.500000,0.994000)(1.600000,0.996800)(1.700000,0.998600)(1.800000,0.998900)(1.900000,0.998900)(2.000000,0.999500)(2.100000,0.999800)(2.200000,0.999900)(2.300000,0.999800)(2.400000,0.999900)(2.500000,1.000000)(2.600000,1.000000)(2.700000,1.000000)(2.800000,1.000000)(2.900000,1.000000)(3.000000,1.000000)(3.100000,1.000000)(3.200000,1.000000)(3.300000,1.000000)(3.400000,1.000000)(3.500000,1.000000)(3.600000,1.000000)(3.700000,1.000000)(3.800000,1.000000)(3.900000,1.000000)(4.000000,1.000000)(4.100000,1.000000)(4.200000,1.000000)(4.300000,1.000000)(4.400000,1.000000)(4.500000,1.000000)(4.600000,1.000000)(4.700000,1.000000)(4.800000,1.000000)(4.900000,1.000000)(5.000000,1.000000)
};
\addlegendentry{$\nu=0.9$}
\addplot[color=red,line width=1.5pt,mark size=2pt,mark=triangle*] coordinates
{
(0.500000,0.618700)(0.600000,0.787700)(0.700000,0.887000)(0.800000,0.950900)(0.900000,0.979500)(1.000000,0.993300)(1.100000,0.997900)(1.200000,0.999400)(1.300000,1.000000)(1.400000,1.000000)(1.500000,1.000000)(1.600000,1.000000)(1.700000,1.000000)(1.800000,1.000000)(1.900000,1.000000)(2.000000,1.000000)(2.100000,1.000000)(2.200000,1.000000)(2.300000,1.000000)(2.400000,1.000000)(2.500000,1.000000)(2.600000,1.000000)(2.700000,1.000000)(2.800000,1.000000)(2.900000,1.000000)(3.000000,1.000000)(3.100000,1.000000)(3.200000,1.000000)(3.300000,1.000000)(3.400000,1.000000)(3.500000,1.000000)(3.600000,1.000000)(3.700000,1.000000)(3.800000,1.000000)(3.900000,1.000000)(4.000000,1.000000)(4.100000,1.000000)(4.200000,1.000000)(4.300000,1.000000)(4.400000,1.000000)(4.500000,1.000000)(4.600000,1.000000)(4.700000,1.000000)(4.800000,1.000000)(4.900000,1.000000)(5.000000,1.000000)
};
\addlegendentry{$\nu=30$}

\end{axis}
\end{tikzpicture}
\end{center}
\caption{Detection probability with respect to $a$, ${\bf p}={\bf a}(\theta)$ with $\theta=20^o$, $N=30$, $n=60$, $\nu=0.1,0.9,30$: K-distributed clutter, false alarm probability =$5\%$.}
\label{fig:effect_nu}
\end{figure}

\section{Conclusion}
In this paper, we address the setting of the regularization parameter when the RSCM or the RTE are used in the ANMF detector statistic as a replacement of the unknown covariance matrix, thereby yielding the schemes ANMF-RSCM and ANMF-RTE. One major bottleneck toward determining the regularization parameter that optimizes the performances of the ANMF detector, is linked to the difficulty to clearly characterize the distribution  of the ANMF statistics under the cases of presence or absence of a signal of interest ($H_1$ and $H_0$). In order to deal with this issue, we considered   the regime under which the number of samples and their dimensions grow large simultaneously. Based on tools from random matrix theory along with recent asymptotic results on the behaviour of the RTE, we derived the asymptotic distribution of the ANMF detector under hypothesis $H_0$ and $H_1$. The obtained results have allowed us to propose an optimal design of the regularization parameter that maximizes the detection probability while keeping fixed the false alarm probability through an appropriate tuning of the threshold value. Simulations results  clearly illustrated the gain of our method over previously proposed empirical settings of the regularization coefficient.
One major advantage of our approach is that, contrary to first intuitions, it leads to simple closed-form expressions that can be easily implemented  in practice. This is quite surprising given that the handling of the classical regime where $n$ grows to infinity with $N$ fixed has been shown to be delicate. As a matter of fact, it has thus far been considered only for the non-regularized Tyler estimator where  intricate expressions in the form of integrals were obtained \cite{pascal-icassp15}. Building the bridge between both approaches is an open  question that deserves investigation. 
\section{Acknowledgments}
Couillet's work is jointly supported by the French ANR DIONISOS project (ANR-12-MONU-OOO3) and the GDR ISIS--GRETSI ``Jeunes Chercheurs'' Project. Pascal's work has been partially supported by the ICODE institute, research project of the Idex Paris-Saclay.
 
\appendices
\section{Proof of Theorem \ref{th:detection}}
The proof of Theorem \ref{th:detection} consists of two steps. First, we study the asymptotic behaviour of the detection probability for fixed $\rho$. Then, by a similar argument to the one considered in \cite{couillet-kammoun-14}, we establish the uniformity of the result over the considered set of $\rho$. 
Assume that the received signal vector ${\bf y}$ is given by:
$$
{\bf y}=\frac{a}{\sqrt{N}}{\bf p} +{\bf x}
$$
with $\|{\bf p}\|^2=N$
and let us write $\sqrt{N}\widehat{T}_N^{RSCM}(\rho)$ as:
$$
\sqrt{N}\widehat{T}_N^{RSCM}(\rho)=\sqrt{N}\frac{{\left|\frac{1}{\sqrt{N}}{\bf y}^*\widehat{\bf R}_N^{-1}(\rho)\frac{\bf p}{\sqrt{N}}\right|}}{\sqrt{\frac{1}{N}{\bf y}^*\widehat{\bf R}_N^{-1}(\rho){\bf y}}\sqrt{\frac{{\bf p}^*\widehat{\bf R}_N^{-1}(\rho){\bf p}}{N}}}.
$$
A close inspection of the expression of $\sqrt{N}\widehat{T}_N^{RSCM}(\rho)$ reveals that the fluctuations will be governed by the numerator ${\bf y}^*\widehat{\bf R}_N^{-1}(\rho)\frac{\bf p}{\sqrt{N}}$ since, from classical results of random matrix theory, we know that  quantities in the denominator exhibit a deterministic behaviour, being well-approximated by some deterministic quantities. In effect,
\begin{equation}
\frac{1}{N}{\bf p}^*\widehat{\bf R}_N^{-1}(\rho){\bf p}- \frac{1}{\rho N} {\bf p}^*{\bf Q}_N(\rho){\bf p}\asto 0,
\label{eq:convergence_1}
\end{equation}
while:
\begin{equation}
\label{eq:convergence_2}
\frac{1}{N}{\bf y}^* \widehat{\bf R}_N^{-1}(\rho){\bf y}- \frac{1}{N\rho}\tr {\bf C}_N{\bf Q}_N(\rho)\asto 0.
\end{equation}
The first convergence \eqref{eq:convergence_1} follows from Theorem 1.1 of \cite{hachem-bilinear-13} whereas the second one is obtained by  observing that, because of the low-SNR hypothesis:
$$
\frac{1}{N}{\bf y}^* \widehat{\bf R}_N^{-1}(\rho){\bf y}-\frac{1}{N}{\bf x}^{*}\widehat{\bf R}_N^{-1}(\rho){\bf x}\asto 0
$$
and then using the well-known convergence result \cite{HAC07}:
$$
\frac{1}{N}{\bf x}^{*}\widehat{\bf R}_N^{-1}(\rho){\bf x}-\frac{1}{N\rho}\tr {\bf C}_N{\bf Q}_N(\rho)\asto 0.
$$
We will now deal with the fluctuations of the numerator. We have:
$$
\sqrt{N}\left|\frac{1}{N}{\bf y}^*\widehat{\bf R}_N^{-1}{\bf p}\right| =
\left|\frac{a{\bf p}^*}{N}\widehat{\bf R}_N^{-1}{{\bf p}}+{\bf x}^*\widehat{\bf R}_N^{-1}\frac{\bf p}{\sqrt{N}}\right|.
$$
Arguing in a similar way to that in \eqref{eq:convergence_1}, we  know that the quantity $\frac{1}{N}{a{\bf p}^*}\widehat{\bf R}_N^{-1}{{\bf p}}$ does not fluctuate and converges to:
$$
\frac{1}{N}{a{\bf p}^*}\widehat{\bf R}_N^{-1}{{\bf p}}-\frac{a}{N\rho} {\bf p}^*{\bf Q}_N(\rho){\bf p}\asto 0,
$$
while, from \cite{couillet-kammoun-14}:
\begin{align*}
&\left[\frac{1}{\sqrt{N}}\Re({\bf x}^*\widehat{\bf R}_N^{-1}{\bf p}),\frac{1}{\sqrt{N}}\Im({\bf x}^*\widehat{\bf R}_N^{-1}{\bf p})\right]^{\mbox{\tiny T}}\\
&-\sqrt{\frac{1}{2\rho^2N}\frac{{\bf p}^*{\bf C}_N{\bf Q}_N^2(\rho){\bf p}}{\left(1-cm(-\rho)^2(1-\rho)^2\frac{1}{N}\tr {\bf C}_N^2{\bf Q}_N^2(\rho)\right)}} Z^{'}=o_p(1)
\end{align*}
for some $Z^{'}\sim\mathcal{N}(0,{\bf I}_2)$.

Let ${\bf r}=\left[\frac{1}{\sqrt{N}}\Re({\bf x}^*\widehat{\bf R}_N^{-1}{\bf p}),\frac{1}{\sqrt{N}}\Im({\bf x}^*\widehat{\bf R}_N^{-1}{\bf p})\right]^{\mbox{\tiny T}}$. Denote by $\Upsilon_N$ and $\omega_N$ the quantities:
\begin{align*}
\Upsilon_N&=\sqrt{\frac{1}{2\rho^2N}\frac{{\bf p}^*{\bf C}_N{\bf Q}_N^2(\rho){\bf p}}{\left(1-cm(-\rho)^2(1-\rho)^2\frac{1}{N}\tr {\bf C}_N^2{\bf Q}_N^2(\rho)\right)}}\\
\omega_N&=\frac{a}{N\rho}{\bf p}^*{\bf Q}_N(\rho){\bf p}
\end{align*}
Recall the following distance between probability distributions:
$$
\beta\left(\mathbb{P},\tilde{\mathbb{P}}\right)=\sup\left\{\int f d\mathbb{P}-\int f d\tilde{\mathbb{P}}, \|f\|_{BL}\leq 1\right\}
$$
where $\|f\|_{BL}=\|f\|_{Lip}+\|f\|_{\infty}$, $\|f\|_{Lip}$ being the Lipschitz norm and $\|.\|_{\infty}$, the supremum norm \cite{dudley}.
 Assume for the moment that $\lim\sup \Upsilon_N<\infty$ and $\lim\sup \frac{a}{N\rho}{\bf p}^*{\bf Q}_N(\rho){\bf p}<\infty$. The proof for these statements will be provided later. Then, from Theorem 11.7.1 in \cite{dudley},
$$
\beta\left(\mathcal{L}\left({\bf r},\frac{1}{N}{a{\bf p}^*}\widehat{\bf R}_N^{-1}{{\bf p}}\right),\mathcal{L}\left(\Upsilon_N Z^{'},\omega_N\right)\right)\to 0.
$$
where $\mathcal{L}(X)$ stands for the probability distribution of $X$. This in particular establishes that the random variable $\left({\bf r},\frac{1}{N}{a{\bf p}^*}\widehat{\bf R}_N^{-1}{{\bf p}}\right)$ converges uniformly in distribution to $\left(\Upsilon_N Z^{'},\omega_N\right)$. From the uniform continuous mapping Theorem in \cite[Theorem 1]{kasy}, we thus prove that 
${\sqrt{N}}\left|\frac{1}{N}{\bf y}^*\widehat{\bf R}_N^{-1}{\bf p}\right|$ behaves asymptotically as a Rice random variable with location  $\frac{a}{N\rho} {\bf p}^*{\bf Q}_N(\rho){\bf p}$ and scale $\sqrt{\frac{1}{2\rho^2N}\frac{{\bf p}^*{\bf C}_N{\bf Q}_N^2(\rho){\bf p}}{\left(1-cm(-\rho)^2(1-\rho)^2\frac{1}{N}\tr {\bf C}_N^2{\bf Q}_N^2(\rho)\right)}}$.
Using this result along with Slutsky Lemma, we conclude that   under $H_1$, ${\widehat{T}_N^{RSCM}(\rho)}$ is also asymptotically equivalent to a Rice random variate but with location $\frac{a}{\sqrt{N}}\frac{\sqrt{{\bf p}^*{\bf Q}_N(\rho){\bf p}}}{\sqrt{\frac{1}{N}\tr{\bf C}_N{\bf Q}_N(\rho)}}$ and scale $\sigma_{N,SCM}$. We therefore get, for a fixed $\rho$,
$$
\left|\mathbb{P}\left[\widehat{T}_N^{SCM}> \frac{r}{\sqrt{N}}|H_1\right] -Q_1\left(g_{SCM}({\bf p}),\frac{r}{\sigma_{N,SCM}}\right)\right| \asto 0
$$ 
The generalization of this result to uniform convergence across $\rho\in\mathcal{R}_{\kappa}^{SCM}$ can be derived along the same steps as in \cite{couillet-kammoun-14}.
We now provide details about the control of the $\lim\sup\Upsilon_N$ and $\lim\sup\omega_N$. The fact that $\lim\sup\omega_N <\infty$ follows directly from  the last item in Assumption \ref{ass:gaussian}, while the control of $\lim\sup \Upsilon_N<\infty$ requires one to check that:
$$
\lim\inf 1-cm(-\rho)^2(1-\rho)^2\frac{1}{N}\tr {\bf C}_N^2{\bf Q}_N^2(\rho).
$$ 
The proof hinges on the observation that this term naturally appears when computing the derivative of  $m(z)$ with respect to $z$ at $z=-\rho$. Simple calculations reveal that:
\begin{align*}
m'(z)&=\left(-z+\frac{c(1-\rho)}{N}\tr {\bf C}_N{\bf Q}_N(z)\right)^{-2}\\
&\times\left(1-m(z)^2(1-\rho)^2\frac{1}{N}\tr {\bf C}_N^2{\bf Q}_N^2(z)\right)^{-1}.
\end{align*}
It suffices thus to show that $m'(-\rho)$ is bounded. As $m$ is a Stieltjes transform of some positive probability measure $\mu$,  it can be written as:
$$
m'(-\rho)=\int\frac{\mu(dx)}{(x+\rho)^2}\leq \frac{1}{\kappa^2}
$$
which ends the proof.

\label{app:detection}
\section{Proof of proposition \ref{prop:frho}}
For ease of notation, we denote by $f(\rho)$ and $\hat{f}(\rho)$, the quantities ${f}_{SCM}(\rho)$ and $\hat{f}_{SCM}(\rho)$.
It is easy to see that $\hat{f}(\rho)$ and $f(\rho)$ converges to an undetermined form as $\rho\uparrow 1$. 
Set $\hat{f}(\rho)\triangleq \frac{\hat{h}(\rho)}{\hat{g}(\rho)}$ with $\hat{g}$ and $\hat{h}$ being given by:
\begin{align*}
\hat{g}(\rho)&=(1-\rho)\left({\bf p}^*\widehat{\bf R}_N^{-1}(\rho){\bf p}\right)^2(1-\rho)\left(1-c+\frac{c}{N}\rho\tr\widehat{\bf R}_N^{-1}(\rho)\right)^2\\
\hat{h}(\rho)&={\bf p}^*\widehat{\bf R}_N^{-1}(\rho){\bf p}-\rho{\bf p}^*\widehat{\bf R}_N^{-2}(\rho){\bf p}
\end{align*}
The handling of the values of $\rho$ approaching $1$ can be performed using the l'Hopital's rule. 

The idea of the proof is to treat seperately the values of $\rho$ in the interval $\left[\kappa,1-\kappa\right]$ and those in $\left[1-\kappa,1\right]$ for some $\kappa$ small enough. In order to allow for a setting of $\kappa$ that is independent from $N$, we need to prove that:
\begin{equation}
\lim_{\rho\uparrow 1}\lim\sup_{N}\left|\hat{f}(\rho)-\frac{h_N^{'}(1)}{g_N^{'}(1)}\right|= 0.
\label{eq:necessary}
\end{equation}
To this end, a uniform variant of the l'Hopital's rule is essential. This variant is stated in the following Lemma:
\begin{lemma}
Let $f_N(\rho)=\frac{h_N(\rho)}{g_N(\rho)}$ with $h_N$ and $g_N$ being defined  in the interval $\rho\in\left[0,1\right]$. Assume that $h_N(1)=g_N(1)=0$ while $\lim\inf_N \left.\frac{dg_N}{d\rho}\right|_{\rho=1} > 0$, $\lim\sup_N \left.\frac{dg_N}{d\rho}\right|_{\rho=1}<+\infty$ and $\lim\sup_N \left.\frac{dh_N}{d\rho}\right|_{\rho=1}<+\infty$. Assume also that the second derivatives of $h_N$ and $g_N$  are uniformly bounded in $N$, that is:
\begin{align*}
&\sup_{\rho\in\left[0,1\right]} \limsup_N \left|h_N^{''}(\rho)\right| <+\infty\\
&\sup_{\rho\in\left[0,1\right]} \limsup_N \left|g_N^{''}(\rho)\right| <+\infty
\end{align*}
 Then,
$$
\lim_{\rho\to 1} \lim\sup_N\left|\frac{h_N(\rho)}{g_N(\rho)}- \frac{h_N^{'}(1)}{g_N^{'}(1)}\right|\to 0.
$$
\label{lemma:intermediate}
\end{lemma}
\begin{proof}
The proof relies on the Talor expansion of $h_N$ and $g_N$ in the vicinity of $1$, which asserts that for any $\rho\in\left[0,1\right]$ there exist $\xi_1$ and $\xi_2$ satisfying:
\begin{align*}
h_N(\rho)&=h_N^{'}(1)(\rho-1)+(\rho-1)^2h_N^{''}(\xi_1)\\
g_N(\rho)&=g_N^{'}(1)(\rho-1)+(\rho-1)^2g_N^{''}(\xi_2)\\
\end{align*}
We therefore have,
\begin{align*}
&\lim\sup_N\left|\frac{h_N(\rho)}{g_N(\rho)}-\frac{h_N^{'}(1)}{g_N^{'}(1)}\right|=\lim\sup_N\left|\frac{h_N^{'}(1)+(\rho-1)h_N^{''}(\xi_1)}{g_N^{'}(1)+(\rho-1)g_N^{''}(\xi_2)}-\frac{h_N^{'}(1)}{g_N^{'}(1)}\right|\\
&=\left|\frac{(\rho-1)h_N^{''}(\xi_1)g_N^{'}(1)-(\rho-1)h_N^{'}(1)g_N^{''}(\xi_2)}{g_N^{'}(1)\left(g_N^{'}(1)+(\rho-1)g_N^{''}(\xi_2)\right)}\right|\\
&\leq |\rho -1|\frac{\limsup_N |h_N^{''}(\xi_1) g_N^{'}(1)|+ \limsup_N |h_N^{'}(\xi_1) g_N^{''}(\xi_2)|}{\lim\inf\left|g_N^{'}(1)\right|^2}
\end{align*}
Tending $\rho$ to $1$ establishes the desired result.
\end{proof}

Obviously functions $\hat{h}(\rho)$ and $\hat{g}(\rho)$ satisfiy the assumptions of Lemma \ref{lemma:intermediate}. Applying  l'Hopital's rule and using the differentiation rules $\frac{d}{d\rho}\widehat{\bf R}_N^{-1}(\rho)=-\widehat{\bf R}_N^{-2}(\rho)\left(-\widehat{\bf R}_N+{\bf I}\right)$ and $\frac{d}{d\rho}\widehat{\bf R}_N^{-2}(\rho)=-2\widehat{\bf R}_N^{-3}(\rho)\left(-\widehat{\bf R}_N+{\bf I}\right)$, we finally prove:
\begin{equation}
\lim_{\rho\uparrow 1}\lim\sup_N\left|\hat{f}(\rho)-\frac{N}{{\bf p}^*\widehat{\bf R}_N{\bf p}}\right|= 0.
\label{eq:essential}
\end{equation}
Now, using the fact $\frac{1}{N}{\bf p}^*\widehat{\bf R}_N{\bf p}-\frac{1}{N}{\bf p}^*{\bf C}_N{\bf p}\asto 0$ in conjunction to the last item in Assumption \ref{ass:gaussian}, we get:
\begin{equation}
\lim_{\rho\uparrow 1}\lim\sup_N\left|\hat{f}(\rho)-\frac{N}{{\bf p}^*{\bf C}_N{\bf p}}\right|\asto 0.
\label{eq:1}
\end{equation}
On the other hand, a careful analysis of the behaviour of $f(\rho)$ near $1$ reveals similarly that:
\begin{equation}
\lim_{\rho\uparrow 1}\lim\sup_N\left|f(\rho)-\frac{N}{{\bf p}^*{\bf C}_N{\bf p}}\right|\to0
\label{eq:2}
\end{equation}
Combining \eqref{eq:1} with \eqref{eq:2}, we finally obtain:
$$
\lim_{\rho\uparrow 1}\lim\sup_N\left|\hat{f}(\rho)-f(\rho)\right|\to0
$$

It then suffices to prove Proposition \ref{prop:frho} on $\mathcal{R}_{\kappa}\triangleq \left[\kappa,1-\ell\right]$.
To this end, we need to recall the following relations satisfied by $m_N(-\rho)$:
$$
m_N(-\rho)=\frac{1-c}{\rho}+\frac{c}{\rho}\frac{1}{N}\tr{\bf Q}_N(\rho)
$$
and
$$
m_N(-\rho)=\left(\rho+c(1-\rho)\frac{1}{N}\tr{\bf C}_N{\bf Q}_N(\rho)\right)^{-1}
$$
Combining these relations, we therefore get:
$$
\frac{1}{N}\tr{\bf C}_N{\bf Q}_N(\rho)=\frac{\rho\left(1-\frac{1}{N}\tr {\bf Q}_N(\rho)\right)}{(1-c)(1-\rho)(1-c+\frac{c}{N}\tr {\bf Q}_N(\rho))}
$$
The result thus follows by using Proposition \ref{prop:estimation_variance} and noticing, in the same way as in \cite{couillet-kammoun-14}, that:
$$
\sup_{\rho\in\left[\kappa,1-\ell\right]}\left|\frac{1}{N}\tr{\bf Q}_N-\frac{1}{\rho}\tr \widehat{\bf R}_N^{-1}(\rho)\right|\asto 0,
$$
$$
\sup_{\rho\in\left[\kappa,1-\ell\right]}\left|\frac{{\bf p}^*}{\sqrt{N}}{\bf Q}_N{\bf p}-\frac{1}{\rho\sqrt{N}}{\bf p}^*\widehat{\bf R}_N^{-1}(\rho){\bf p}\right|\asto 0
$$
and
\begin{align}
&\sup_{\rho\in\left[\kappa,1-\ell\right]}\left|\frac{\frac{1}{N}{\bf p}^*{\bf C}_N{\bf Q}_N^2(\rho){\bf p}}{1-cm(-\rho)^2(1-\rho)^2\frac{1}{N}\tr{\bf C}_N^2{\bf Q}_N^2(\rho)}\right.\\
&-\left.\frac{\frac{1}{N}\left({\bf p}^*\widehat{\bf R}_N^{-1}{\bf p}-\rho{\bf p}^*\widehat{\bf R}_N^{-2}{\bf p}\right)}{(1-\rho)\left(\frac{1-\rho}{c}+c\frac{1}{N}\tr\widehat{\bf R}_N^{-1}(\rho)\right)}\right|\asto 0.
\end{align}
\label{app:frho}
\bibliographystyle{IEEEbib}
\bibliography{IEEEabrv,IEEEconf,./tutorial_RMT.bib}

\end{document}